\documentclass[journal,draftclsnofoot,onecolumn,12pt,twoside]{IEEEtranTCOM}

\IEEEoverridecommandlockouts

\usepackage{epsfig,rotating,setspace,latexsym,amsmath,epsf,amssymb,amsfonts,bm,theorem,cite,algorithm,graphicx,epsf,authblk,epstopdf,color,algpseudocode,bbm}

\usepackage{algorithm} 
\usepackage{algpseudocode} 
\usepackage{subfig}
\usepackage{mathtools}
\usepackage{array}
\DeclarePairedDelimiter\ceil{\lceil}{\rceil}

\newtheorem{remark}{Remark}

\newtheorem{corollary}{Corollary}

\newtheorem{theorem}{Theorem}
\newtheorem{lemma}{Lemma}

\allowdisplaybreaks

\normalsize

\ifCLASSINFOpdf

\else

\fi

\begin{document}

\title{Delay Sensitive Hierarchical Federated Learning with Stochastic Local Updates}

\author{Abdulmoneam~Ali,~\IEEEmembership{Graduate Student Member,~IEEE,}
        and~Ahmed~Arafa,~\IEEEmembership{Member,~IEEE,}
\thanks{The authors are with the Department
of Electrical and Computer Engineering, University of North Carolina at Charlotte,
NC 28223, USA. e-mails: aali28@charlotte.edu, aarafa@charlotte.edu.}
\thanks{This work was supported by the U.S. National Science Foundation under Grants CNS 21-14537 and ECCS 21-46099, and has been presented in part at the Asilomar Conference on Signals, Systems, and Computers, Pacific Grove, CA, USA, October 2022 \cite{ourasilomar}.}}
\date{} 

\maketitle

\vspace{-.75in}
\begin{abstract}
The impact of local averaging on the performance of federated learning (FL) systems is studied in the presence of {\it communication delay} between the clients and the parameter server. To minimize the effect of delay, clients are assigned into different groups, each having its own {\it local} parameter server (LPS) that aggregates its clients' models. The groups' models are then aggregated at a {\it global} parameter server (GPS) that only communicates with the LPSs. Such setting is known as {\it hierarchical} FL (HFL). Unlike most works in the literature, the number of local and global communication rounds in our work is randomly determined by the (different) delays experienced by each group of clients. Specifically, the number of local averaging rounds \textcolor{black}{is tied to} a wall-clock time period coined the {\it sync time} $S$, after which the LPSs synchronize their models by sharing them with the GPS. Such sync time $S$ is then reapplied until a global wall-clock time is exhausted.

First, an upper bound on the deviation between the updated model at each LPS with respect to that available at the GPS is derived. This is then used as a tool to derive the convergence analysis of our proposed delay-sensitive HFL algorithm, first at each LPS \textcolor{black}{individually, and }then at the GPS. Our theoretical convergence bound showcases the effects of the whole system's parameters, including the number of groups, the number of clients per group, and the value of $S$. Our results show that the value of $S$ should be carefully chosen, especially since it implicitly governs how the delay statistics affect the performance of HFL in situations \textcolor{black}{where training time is restricted.}

\end{abstract}


\IEEEpeerreviewmaketitle

%===
\section{Introduction}

Federated Learning (FL) is a distributed machine learning training system in which edge devices (clients) collaboratively train a model of interest based on their locally stored datasets. A central node (parameter server) orchestrates the learning process by collecting the clients' parameters for aggregation \cite{pmlr-v54-mcmahan17a}. Due to its data privacy preserving and bandwidth saving nature, FL has attracted a lot of attention and has been used in diverse applications including healthcare and mobile services. 

\noindent\textbf{Challenges and Related Work.} In order to successfully deploy FL in communication networks, lots of challenges should be addressed. These include: the computing capabilities of the clients; the communication overhead between the clients and the parameter server; and the system heterogeneity, whether in the clients' communication channels or their data statistics. 

\textcolor{black}{One way to overcome the resource-constrained capabilities of clients and the limited channel bandwidth is to employ quantization, sparsification, and compression when the size of the learning model is too large}\cite{bouzinis2022wireless, sattler2019robust}. Another challenge often addressed in the literature is related to the limited available spectrum that hinders the simultaneous participation of all clients. For that, client scheduling and its consequences on the system's performance becomes crucial \cite{cho2021client,wang2022a,AoI}.   

Among all the challenges, communication remains to be the bottleneck issue, and various solutions have been proposed in the literature to mitigate it. One of these solutions is to conduct several local updates at the clients' side before communicating with the parameter server \cite{stich2018local,woodworth2020local,lin2018don,pmlr-v130-shokri-ghadikolaei21a}. Another solution is to introduce intermediate parameter servers, denoted as local parameter servers (LPSs), between the clients and the (now) global parameter server (GPS). Such setting of FL is known in the literature as the \textcolor{black}{\it hierarchical} FL (HFL) setting \cite{wang2022demystifying}. The main advantage of having LPSs close to clients is to reduce the latency and \textcolor{black}{energy required} to communicate with the GPS \cite{Hfl_kh}. In \cite{luo2020hfel}, a joint resource allocation and client association problem is formulated
in an HFL setting and then solved by an iterative algorithm. Moreover, the authors in \cite{delay_aware_Nicolo} utilize the bottleneck time at the GPS by allowing the devices to perform extra local updates while they are waiting for the new global model from the GPS. To overcome the dispersion between the weights evaluated based on the latest global model and those at the stale model, the authors propose a linear
global-local model combiner scheme to address this challenge. 

In these mentioned works, the authors analyze their systems while assuming a fixed number of local iterations and global communication rounds. \textcolor{black}{However, in more realistic scenarios, the number of local iterations may vary from one global communication round to another, depending on the dynamic nature of the (wireless) communication channel and the different computational capabilities of the edge devices.} \textcolor{black}{Moreover, the number of global communication rounds can also vary if the training time is restricted.} One scenario in which this is the case is when model training has to be conducted during non-congested periods of the network.

\noindent\textbf{Contributions.}
Motivated by the aforementioned endeavors and to cope with the very low latency service requirements in 6G networks (and beyond), in this paper we focus on HFL for \textcolor{black}{delay-sensitive} communication networks. We study FL settings that have an additional requirement of conducting training within a predefined deadline. Such scenario is relevant for, e.g., energy-limited clients whose availability for long times is not always guaranteed. To enforce the system to abide by this constraint, the number of local training updates will be determined by a wall-clock time. Specifically, we define a \textcolor{black}{\it sync time} $S$ within which the LPSs are allowed to aggregate the parameters they receive from their groups' clients. Each local iteration consumes a random group-specific \textcolor{black}{delay}, and hence the total number of local updates within $S$ will also be random, and could possibly be \textcolor{black}{different} across groups. This dissimilarity in the delay statistics is introduced to capture, e.g., the effects of wireless channels and different computational resources among different group clients. Following the deadline $S$, the LPSs forward their models to the GPS.

We set another time constraint at the GPS denoted the \textcolor{black}{\it system time} $T$. This is the total time allowed for the overall HFL system to perform the training. Different values of $S$ and $T$ will lead to a different number of local and global updates. By controlling $S$, we also control how many times the clients communicate with the GPS: \textcolor{black}{more local iterations lead to fewer global ones and vice versa.} \textcolor{black}{This is  different from the existing works that assume that the global communication rounds are predefined independently from the number of local updates.} Thus, one of the main questions we address in this work is:

\begin{center}
\textcolor{black}{\it When would it be preferable for an HFL system to perform more local iterations than global ones, with local iterations being relatively faster yet leading to possibly lower model accuracy, and vice versa?}
\end{center}

We present a thorough theoretical convergence analysis for the proposed HFL setting for non-convex loss functions. Our results show how the different system parameters affect the accuracy, \textcolor{black}{namely the wall clock times} $S$ and $T$, the number of \textcolor{black}{groups,} and the number of clients per group. Various experiments are performed on different models and datasets to show how to optimize the sync time $S$ based on the other system parameters. Our system model is depicted in Fig.~\ref{fig:sys_mod_hfl}.

Our work relaxes the assumption in \cite{wang2022demystifying} regarding the non-existence of  communication constraints. In particular, different from \cite{wang2022demystifying}, we incorporate the delay induced by communication and computation processing between LPSs and their associated clients in addition to that in between LPSs and the GPS. Moreover, we point out the trade-off between the number of local iterations and global communication rounds, and  propose an approach to balance between them by optimizing the sync time $S$. We  also add another lens when comparing HFL with centralized FL (which is different from the results in \cite{Hfl_kh}). Specifically, depending on the system parameters and delay statistics, we show that there are scenarios in which direct communication with the GPS (centralized FL) can outperform communicating with the LPSs (HFL).

We now summarize our main contributions and novelty:
\begin{itemize}
    \item To the best of our knowledge, this is the first work that considers stochastic local and global iterations in HFL (as opposed to being predefined prior to training); the number of iterations is determined by the induced delay between the LPSs and their associated clients, in addition to that in between the LPSs and the GPS. 
    \item We show that there is an intrinsic trade-off between local and global iterations in delay-sensitive HFL settings; the more time you spend in local iterations, the less you spend in global ones, and vice versa. Hence, we provide insights on how to characterize the effect of such a trade-off on the training accuracy in HFL. This is captured by carefully selecting the proposed sync time $S$ across LPSs.
    \item Our theoretical results characterize the consequences of considering a wall-clock training time, together with random local iterations and global communication rounds, on the training accuracy of non-convex loss functions in HFL settings. 
    \item Instead of relying on similarity and dissimilarity assumptions about the updated models at the LPSs and the GPS (as done in most of the literature), we derive our own bound to capture the divergence between them as a function of the system parameters. Such a bound plays an instrumental role in proving our convergence results and reveals insights about how the system parameters can be designed to alleviate heterogeneity between groups.
    \item In addition to theoretical analyses, we carry out extensive experiments on various datasets and different scenarios to show the effects of choosing the sync time $S$, along with the remaining system parameters, on delay-sensitive HFL settings.
\end{itemize}

\begin{figure}[t]
\centering
\includegraphics[scale=1]{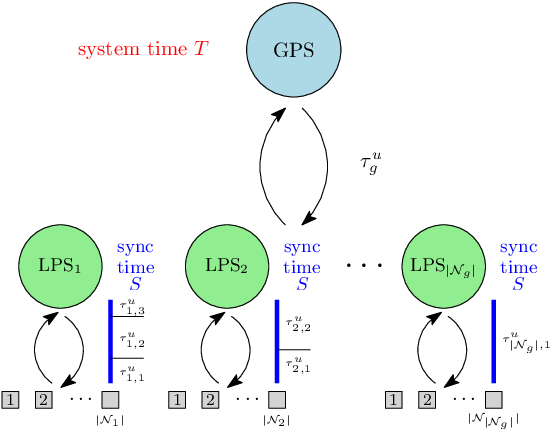}
\caption{System model of delay sensitive HFL.}
\label{fig:sys_mod_hfl}
\end{figure}

\noindent\textbf{Notation and Organization.} $\mathbb{R}$ denotes the real number field; $ \left\|\cdot\right\|$ denotes the Euclidean norm; $\langle x, y \rangle$ denotes the inner product between two vectors $x$ and $y$; $\mathbb{E}$  denotes statistical expectation,  while $\mathbb{E}_{|x}\left\|\cdot\right\|$ represents the conditional expectation given $x$.

The remainder of the paper is organized as follows. Section~\ref{Sys_Model} presents the system model and our proposed HFL algorithm. Theoretical convergence analyses are derived in Section~\ref{main_res}, and verified via extensive simulations in Section~\ref{experiments}. Section~\ref{conclusion} concludes the paper.

\section{System Model}\label{Sys_Model}

We consider an HFL system with a global PS (GPS) and a set of local PSs (LPSs), $\mathcal{N}_g$, that serve a number of clients, see Fig.~\ref{fig:sys_mod_hfl}. Clients are distributed across different LPSs to form clusters (groups), in which a client can only belong to one group, and may only communicate with one specific LPS. 
Denoting by $\mathcal{N}_i$ the set of clients in group $i$, the total number of clients in the system is $\sum_{i \in \mathcal{N}_g}|\mathcal{N}_i|$, with $|\cdot|$ denoting cardinality. Each client has its own dataset, and the data is independently and identically distributed (i.i.d.) among clients. The empirical loss function at the LPS of group $i \in \mathcal{N}_g$ is defined as follows:
\begin{align}
\label{group_loss}
    f_{i}(x)\triangleq\frac{1}{|\mathcal{N}_i|}\sum_{k \in \mathcal{N}_i}F_{i,k}(x), \quad  i \in \mathcal{N}_g,
\end{align}
where $F_{i,k}(x)$ is the loss function at client $k$ in group $i$. The goal of the HFL system is to minimize a global loss function:
\begin{align} \label{global_loss}
    f(x)&\triangleq\frac{1}{\sum_{i \in \mathcal{N}_g} |\mathcal{N}_i|} \sum_{i \in \mathcal{N}_g} |\mathcal{N}_i| f_{i}(x) 
    =\frac{1}{\sum_{i \in \mathcal{N}_g} |\mathcal{N}_i|}\sum_{i \in \mathcal{N}_g} \sum_{k \in \mathcal{N}_i}F_{i,k}(x).
\end{align}

The global loss function is minimized over a number of \textcolor{black}{ global communication rounds} between the GPS and the LPSs. At the beginning of the $u$th global round, the GPS broadcasts the global model, $x^u \in \mathbb{R}^{d} $, with $d$ representing the model dimension, to the LPSs. The LPSs then forward $x^u$ to their associated clients, which is used to run a number of SGD steps based on their own local datasets. After each SGD step, the clients share their models with their LPS, which aggregates them and broadcasts them back locally to its clients. We call this local round trip a \textcolor{black}{ local iteration}. We further illustrate how the global rounds and local iterations interact as follows. Let $x_i^{u,l}$ denote the model available at LPS $i$ after local iteration $l$ during global round $u$, and let $x_{i,k}^{u,l}$ denote the corresponding local model of client $k$ of group $i$. We now have the following equations that build up the models:
\begin{align}
&x_{i}^{u,0}=x^{u},\quad \forall i\in\mathcal{N}_g, \label{eq_lps-model-initial} \\
&x_{i,k}^{u,0}=x_{i}^{u,0},\\
&x_{i,k}^{u,l}=x_{i}^{u,l-1} - \alpha \: \Tilde{g}_{i,k}\left(x_{i}^{u,l-1}\right), \forall k \in \mathcal{N}_{i}, \label{eq_lps-model-itr}
\end{align}
where $\alpha$ is the learning rate, and $\Tilde{g}_{i,k}$ is an unbiased stochastic gradient evaluated at $x_{i}^{u,l-1}$. After the $l$th SGD step, LPS $i$ collects $\left\{x_{i,k}^{u,l}\right\}$ from its associated clients and aggregates them to get the $l$th local model,
\begin{align}
    x_{i}^{u,l} =\frac{1}{|\mathcal{N}_i|} \sum_{k \in \mathcal{N}_i} x_{i,k}^{u,l}, \label{eq_lps-model-agg}
\end{align}
which is shared with its clients to initialize SGD step $l+1$.
Each local iteration takes a \textcolor{black}{\it random} time to be completed. This includes the time for broadcasting the local model by the LPS to its clients, the SGD computation time, and the aggregation time. Let $\tau_{i,l}^u$ denote the wall-clock time elapsed during local iteration $l$ for group $i$ in global round $u$. The time $\tau_{i,l}^u$ encompasses the training, communication, and aggregation time between the LPS and its associated clients. We assume that $\tau_{i,l}^u$'s are i.i.d. across local iterations $l$ and global rounds $u$, but may not be identical across groups $i$. This is motivated by the different channel delay statistics that each group may experience when communicating with its LPS. In addition to that, each group may have clients with heterogeneous computational capabilities. These two factors together hinder one group to (statistically) do an identical number of local updates like other groups. We define a {\it sync time,} $S$, that represents the allowed local training time for {\it all} groups. After the sync time, the LPSs need to report their local models to the GPS, and thereby ending the current global round. During global round $u$, and within the sync time $S$, group $i$ will then conduct a random number of local iterations given by 
\begin{align}\label{num_local_iters}
    t_{i}^{u}\triangleq\min\left\{n:~\sum_{l=1}^{n} \tau_{i,l}^u \geq S\right\},\quad i \in \mathcal{N}_g.
\end{align}
Observe that the statistics of $t_i^u$'s are not identical across groups, see Fig.~\ref{fig_s-protocol-example} for an example sample path during global round $u$.  Based on (\ref{num_local_iters}), it is clear that the value of $S$ determines the number of local iterations that each group conducts.  After the $t_i^u$  \textcolor{black}{local iterations are done}, and using (\ref{eq_lps-model-initial})--(\ref{eq_lps-model-agg}), LPS $i$ will have acquired the following model:
\begin{align} \label{eq_local-update}
        x_{i}^{u,t_{i}^{u}} = x_{i}^{u,0}-\frac{\alpha}{|\mathcal{N}_i|} \sum_{l=0}^{t_{i}^{u}-1}\sum_{k \in \mathcal{N}_i} \Tilde{g}_{i,k}\left(x_{i}^{u,l}\right).
\end{align}
We consider a synchronous setting in which the GPS waits for all the LPSs to finish their local iterations before a global aggregation. Since LPSs incur different wall-clock times to collect their models, some of them may need to stay \textcolor{black}{idle,} waiting for others to finish. The GPS therefore starts aggregating the models after
\begin{align}
\max_{i\in\mathcal{N}_g}\left\{\sum_{l=1}^{t_i^{u}} \tau_{i,l}^{u}\right\}
\end{align}
time units from the start of the local iterations in global round $u$. We denote this \textcolor{black}{period by} the \textcolor{black}{syncing period} (see Fig.~\ref{fig_s-protocol-example}),  which captures the straggler's effect in synchronous FL. When updating the GPS, LPS $i$ sends the difference between its final and initial models divided by the number of its local iterations performed \cite{fedvarp}, \cite{why_divison}, i.e., it sends
\begin{align} \label{local_update}
        \frac{1}{t_i^u}\!\left(\!x_{i}^{u,t_{i}^{u}}\!-\! x_{i}^{u,0}\!\right)\!=\frac{-\alpha}{|\mathcal{N}_i|t_i^u} \sum_{l=0}^{t_{i}^{u}-1}\! \!\sum_{k \in \mathcal{N}_i} \Tilde{g}_{i,k}\left(x_{i}^{u,l}\right),i \in \mathcal{N}_g.
\end{align}
We note that the purpose of diving by $t_i^{u}$ is to avoid \textit{biasing} the \textcolor{black}{global model} and to avoid the objective inconsistency \cite{why_divison}. Moreover, normalization by the number of local iterations forces the aggregated model update to be the result of an \textit{equal} contribution from all groups. To see this, observe that (cf. Assumption~2) 
\begin{align}
 \mathbb{E}_{|\bm t_i^u} \frac{1}{t_i^u}\left(x_{i}^{u,t_{i}^{u}}- x_{i}^{u,0}\right) &=\frac{-\alpha}{|\mathcal{N}_i|t_i^u}\sum_{l=0}^{t_{i}^{u}-1}\sum_{k \in \mathcal{N}_i} \!\!\nabla F_{i,k}\!\left(x_{i}^{u,l}\right) 
 =\frac{-\alpha}{t_i^u}\sum_{l=0}^{t_{i}^{u}-1}   \nabla  f_{i}\left(x_{i}^{u,l}\right), %\nonumber \\
 % &= -\alpha \nabla  f_{i}(x_{i}^{u,l}),
\end{align}
where $\mathbb{E}_{|\bm t_i^u} $ \textcolor{black}{denotes the} conditional expectation given the vector ${\bm t}_i^u\triangleq\left\{t_{i}^{u^\prime}\right\}_{u^\prime=1}^{u}$.
\begin{figure}[t]
\centering
\includegraphics[width=0.95\linewidth]{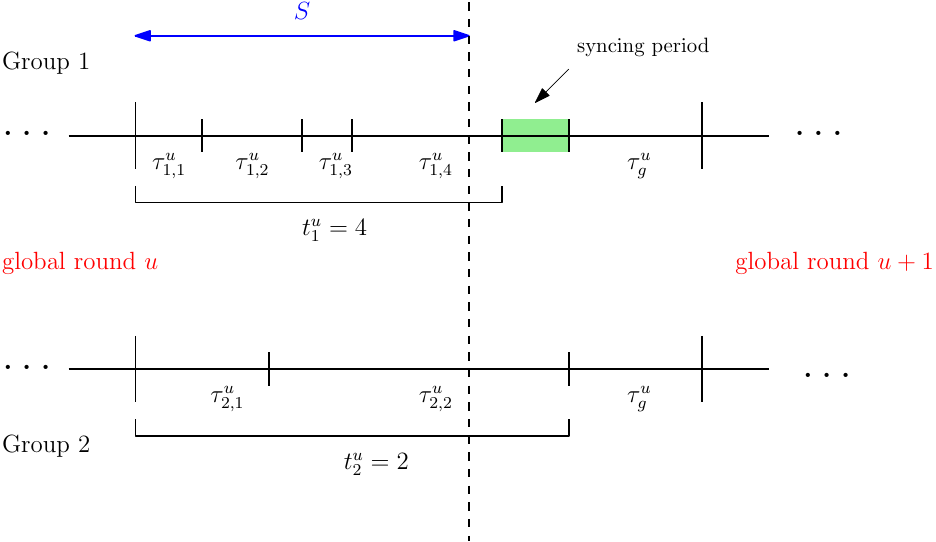}
\caption{Example sample path of global rounds and local iterations of 2 groups with wall-clock times considerations.}
\label{fig_s-protocol-example}
% \vspace{-.45in}
\end{figure}
The GPS then updates its global model as 
\begin{align}\label{global_update}
         x^{u+1}&=x^{u}+\sum_{i \in \mathcal{N}_g}\frac{|\mathcal{N}_i|}{\sum_{i \in \mathcal{N}_g} |\mathcal{N}_i|}\frac{1}{t_{i}^{u}}\left(x_{i}^{u,t_{i}^{u}}- x_{i}^{u,0}\right)
         =x^{u}-\frac{\alpha}{\sum_{i \in \mathcal{N}_g}|\mathcal{N}_i|}\sum_{i \in \mathcal{N}_g} \!\frac{1}{t_{i}^{u}} \sum_{l=0}^{t_{i}^{u}-1}\!\sum_{k \in \mathcal{N}_i}\! \!\Tilde{g}_{i,k}\left(x_{i}^{u,l}\right), 
\end{align}
and broadcasts $x^{u+1}$ to the LPSs to begin global round $u+1$. We assume that the global aggregation and broadcasting processes consume i.i.d. $\tau_g^u$'s wall-clock times. An example of the HFL setting considered is depicted in Fig.~\ref{fig_s-protocol-example}.
The overall HFL training process stops after a total \textcolor{black}{system time} $T$. The value of $T$ represents \textcolor{black}{the time budget allowed for training in applications sensitive to delays.} Within $T$, the total number of global rounds will be given by
\begin{align}\label{num_global_iters}
    \mathcal{U}\triangleq\min\left\{m:~\sum_{u=1}^{m} \left(\max_{i\in\mathcal{N}_g}\left\{\sum_{l=1}^{t_i^{u}}\tau_{i,l}^{u}\right\} + \tau_{g}^{u} \right) \geq T\right\}.
\end{align}
Notably in (\ref{num_global_iters}), the number of global iterations $\mathcal{U}$, will be dominated by the number of local iterations $\{{t_i^{u}}\}$, the delay at the GPS $\tau_{g}^{u}$, and the training time constraint $T$. We coin the proposed algorithm {\it delay sensitive HFL} which is summarized in Algorithm~\ref{alg_main}. We also include a list of the key notations in Table~\ref{tab:key_notation} for better presentation.

In the sequel, we analyze the performance of Algorithm~\ref{alg_main} in terms of the \textcolor{black}{wall-clock times,} number of clients and other parameters of the system. We then discuss how to choose the sync time $S$ to \textcolor{black}{ensure better learning outcomes}.

\begin{table}[h]

\captionsetup{labelfont=normalfont,textfont=normalfont}
\centering
\begin{tabular}{| m{1.3cm} | m{14.3cm}|}%{|c|c{5cm} |}
\hline
Symbol & Definition \\
\hline
$\mathcal{N}_g$ & Set of local PSs (LPSs) \\ \hline
$\mathcal{N}_i$ & Set of clients in group $i$ \\ \hline
$t_i^u$ &  Number of local iterations (averaging) 
performed by group $i$ within global round $u$ \\  \hline
$x^u$ & Global model after $u$ global iterations \\ \hline
 $x_{i,k}^{u,l}$ & Local model of client $k$ of group $i$ at local iteration $l$ within global round $u$  \\ \hline
$\tau_{i,l}^u$ & Wall-clock time elapsed during local iteration $l$ for group $i$ in global round $u$ \\ \hline
$\tau_g^{u}$ & Wall-clock time elapsed at the GPS to update the global model in global round $u$ \\ \hline
$S$ & Sync time; allowed local training time for
all groups\\ \hline
$T$ & System time; overall allowed training time \\ \hline
$\mathcal{U}$ & Number of global iterations \\
\hline
\end{tabular}
\caption{Key notations and system parameters.}
\label{tab:key_notation}

\end{table}

\begin{algorithm}[t]
	\caption{Delay Sensitive HFL} 
	\begin{algorithmic}[1]
%	\While{$\Bar{T} \leq T$}{kk}
\State \textbf{Input:} learning rate $\alpha$, system time $T$, sync time $S$
\State \textbf{Output:} global aggregated model $x^{\mathcal{U}}$
\State \textbf{Initialization:} $\Bar{T},u \gets 0$
\While {$\Bar{T} \leq T$}
		\State \underline{Global Broadcast:} $x_i^{u,0} \gets x^{u}, \: \forall i \in \mathcal{N}_g$
		\For {$i \in \mathcal{N}_g$}
				\State $t_i^{u}\gets 0 , \Bar{t} \gets 0$
		\While {$\Bar{t} \leq S$}
			\For {$k \in \mathcal{N}_i$}
				\State \underline{SGD Update:} $x_{i,k}^{u,l} \!=x_{i}^{u,l-1} - \!\alpha  \Tilde{g}_{i,k}(x_{i}^{u,l-1})$
			\EndFor
			\State \underline{Local Aggregation:}     $x_{i}^{u,l} =\frac{1}{|\mathcal{N}_i|} \sum_{k \in \mathcal{N}_i} x_{i,k}^{u,l} $
			\State \underline{Group Broadcast:} $x_{i,k}^{u,l}=x_{i}^{u,l}$
			\State \underline{Local updates increment:} $t_i^{u} \gets t_i^{u}+1$, \\
                                                \hspace{.625in}$\Bar{t} \gets \Bar{t}+\tau_{i}^{u}$
			\EndWhile
			\State \underline{Upload:} $\frac{1}{t_{i}^{u}}(x_{i}^{u,t_{i}^{u}}- x_{i}^{u,0})$
		\EndFor
		\State \underline{Global Update:} \\
            \hspace{.25in}$x^{u+1}=x^{u}+\sum_{i \in \mathcal{N}_g}\frac{|\mathcal{N}_i|}{\sum_{j \in \mathcal{N}_g} |\mathcal{N}_j|}\frac{1}{t_{i}^{u}}(x_{i}^{u,t_{i}^{u}}- x_{i}^{u,0})$ 

		\State \underline{System Time Update:} $\Bar{T} \gets \max_{i}\{\sum_{j=1}^{(t_i^{u})} \tau_{i,j}^{u}\}+ \tau_{g} $,\\
        \hspace{.25in}$u \gets u+1$
		\EndWhile
	\end{algorithmic} \label{alg_main}
\end{algorithm}

\section{Main Results}\label{main_res}
In this section, we present the convergence analysis for the proposed HFL setting. We have the following typical assumptions about the loss function and SGD \cite{Hfl_kh}:

\noindent\textbf{Assumption 1.} (\textit{Smoothness}). Loss functions are  $L$-smooth: $\exists L > 0$ such that $\forall x,y \in \mathbb{R}^d$: 
\begin{align} \label{assum_1}
    \hspace{-0.5cm}F_{i,k}(y) \leq F_{i,k}(x)+\langle \nabla F_{i,k}(x),y-x \rangle+\frac{L}{2} \left\| y-x\right\|^2,\forall i,k.
\end{align}
\textbf{Assumption 2.} (\textit{Unbiased Gradient}). The gradient estimate at each client satisfies
\begin{align} \label{assum_2}
    \mathbb{E}\Tilde{g}_{i,k}(x)= \nabla F_{i,k}\left(x\right),\quad \forall i,k.
\end{align} 
 \textbf{Assumption 3.} (\textit{Bounded Gradient}). There exists a constant $G >0 $ such that the stochastic gradient's second moment is bounded as
\begin{align}\label{assum_3}
\mathbb{E}\left\|\Tilde{g}_{i,k}(x)\right\|^2 \leq G^2,\quad \forall i,k.
\end{align}
  \textbf{Assumption 4.} (\textit{Bounded Variance}). There exists a constant $\sigma >0 $, such that the variance of the stochastic gradient is bounded as
\begin{align} \label{assum_4}
    \mathbb{E}\left\| \Tilde{g}_{i,k}(x)-\nabla F_{i,k}(x)\right\|^2 \leq \sigma^2,\quad \forall i,k.
\end{align}

It is worth noting that we conduct our analysis {\it without} assuming convexity of the loss function at any entity in the system. According to our proposed algorithm, after each global round, the group clients will resume their local training from the aggregated global model instead of the their latest local one. Hence, we need to quantify the {\it deviation} between the two parameter models through the following lemma (the proof is in Appendix~\ref{appB}):
\begin{lemma} \label{lemma_1}
For $0 \leq \alpha \leq \frac{1}{L}$, the delay sensitive HFL algorithm satisfies the following $\forall u, i$:
\begin{align} 
       \mathbb{E}_{|\bm t_i^u} \!\left\| x^{u+1}-x_{i}^{u,t_{i}^{u}}\right\|^2 \!\!\leq 2\alpha^2 \left( \left( \!t_{i}^{u}\right)^2\!+ \frac{|\mathcal{N}_g|\sum_{j \in \mathcal{N}_g }|\mathcal{N}_j|^2}{(\sum_{i \in \mathcal{N}_g}|\mathcal{N}_i|)^2} \!\right)\!G^2. \label{eq_divergence}
\end{align}

\end{lemma}

\begin{remark}
The first term in the bound in Lemma~\ref{lemma_1} represents the contribution of group $i$ while the second one reflects the impact of all groups in the deviation between the parameter models. It is obvious that more local iterations lead to more deviation between the local and the global models. Note that local iterations are the sole determinant of the deviation in case of having one group only (e.g., when there is no hierarchy); having two or more groups carries an additive effect on the deviation as seen in the second term. 
\end{remark}
\begin{remark}
In case of having only one group in the system, one gets a strictly smaller upper bound than that in \cite{yu2019parallel}, which is given by $4\alpha^2 \left(t_{i}^{u}\right)^2 G^2$ (almost two times the bound in (\ref{eq_divergence}) for $|\mathcal{N}_g|=1$ for large values of $t_i^u$).
\end{remark}
\vspace{-.3in}
\begin{remark}
The lemma suggests a way to overcome the effects of the heterogeneity of the number of local iterations conducted by each group on the divergence. By setting a different learning rate for each group that is inversely proportional to the number of local iterations (e.g., $\alpha^2 \leq \frac{1}{2 ((t_i^u)^2+|\mathcal{N}_g|) G^2 }$ for group $i$), one can achieve the same divergence bound ($\leq1$) across all groups.
\end{remark}

Lemma~\ref{lemma_1} serves as a building block for our main convergence theorems of the proposed delay sensitive HFL. These are mentioned next (with proofs in Appendices \ref{appC} and \ref{appD}).
\begin{theorem}[\textbf{Convergence  Analysis per Group}]
\label{CA_Group}
For $0 \leq \alpha \leq \frac{1}{L}$, the delay sensitive HFL algorithm achieves the following group $i$ bound for a given $\mathcal{U}$:
\begin{align}
\frac{1}{\sum_{u=1}^{\mathcal{U}} t_{i}^{u}} \sum_{u=1}^{\mathcal{U}}  &\sum_{l=1}^{t_{i}^{u}} \mathbb{E}_{|\bm t_i^u}\left\|\nabla f_{i}(x_{i}^{u,l-1})\right\|^2 
     \leq \frac{2}{\alpha \sum_{u=1}^{\mathcal{U}} t_{i}^{u}} \Biggl(\mathbb{E}_{|\bm t_i^u}f_{i}\left(x_{i}^{1}\right)-\mathbb{E}_{|\bm t_i^u}f_{i}\left(x_{i}^{\mathcal{U},t_i^{\mathcal{U}}}\right) \Biggl)+ \frac{\alpha  L \sigma^2}{|\mathcal{N}_{i}|}  \nonumber \\
    &+\Biggl(\frac{1}{\alpha \sum_{u=1}^{\mathcal{U}} t_{i}^{u}} +  \frac{2 (L +1) \kappa \alpha}{ \sum_{u=1}^{\mathcal{U}} t_{i}^{u}} \Biggr)(\mathcal{U}-1)  G^2 +\frac{2 (L +1)\alpha}{ \sum_{u=1}^{\mathcal{U}} t_{i}^{u}}   \sum_{u=1}^{\mathcal{U}-1}   ( t_{i}^{u})^2 G^2,
\end{align}

where the term $\kappa$ is given by
$
\kappa\triangleq\frac{|\mathcal{N}_g|}{\left(\sum_{i \in \mathcal{N}_g}|\mathcal{N}_i|\right)^2} \sum_{j \in \mathcal{N}_g} |\mathcal{N}_j|^2.$

\end{theorem}

\begin{remark} \label{iso_grp_convg}
  Notably, setting $\mathcal{U}=1$ means that the groups will work \textit{individually}. The result of Theorem~\ref{CA_Group} shows that convergence is still guaranteed in this isolated case by choosing $0\leq\alpha\leq\min\left\{\frac{1}{L},\frac{1}{\sqrt{t_i^1}}\right\}$.  
\end{remark}
%%%%%%%%%%%%%%%%%%%%
\begin{theorem}[\textbf{Global Convergence  Analysis} ]
\label{global_convg}
For $0 \leq \alpha \leq \frac{1}{L}$, the delay sensitive HFL algorithm achieves the following global bound for a given $\mathcal{U}$:

\begin{align} \label{eq_global_bound} 
&\frac{1}{\mathcal{U}} \sum_{u=1}^{\mathcal{U}} \mathbb{E}_{|\bm t_i^u}\left\|\nabla f\left(x^{u}\right)\right\|^2  \!\leq\!  \frac{2}{\alpha} \frac{1}{\mathcal{U}}  \left(\mathbb{E}_{|\bm t_i^u}f\left(x^{1}\right)\!-\!\mathbb{E}_{|\bm t_i^u}f\left(x^{\mathcal{U}+1}\right)\right)  
+ \frac{\alpha L |\mathcal{N}_g| \sum_{i \in \mathcal{N}_g} |\mathcal{N}_i|^2\sigma^2}{\left(\sum_{i \in \mathcal{N}_g}|\mathcal{N}_i|\right)^2}\nonumber \\
&\hspace{.5in}+ \frac{1}{\mathcal{U}} \sum_{u=1}^{\mathcal{U}}  \frac{12\alpha^2  L^2 |\mathcal{N}_{g}|}{\left(\sum_{i \in \mathcal{N}_g}|\mathcal{N}_i|\right)^2} \sum_{i \in \mathcal{N}_g} |\mathcal{N}_{i}|^2   \left(t_{i}^{u-1}\right)^2 +\frac{12 \alpha^2 L^2 G^2 |\mathcal{N}_{g}|^2}{\left(\sum_{i \in \mathcal{N}_g}|\mathcal{N}_i|\right)^4}   \left(\sum_{i \in \mathcal{N}_g} |\mathcal{N}_{i}|^2\right)^2 
 \nonumber \\
&\hspace{.5in}+ \frac{1}{\mathcal{U}} \sum_{u=1}^{\mathcal{U}}\frac{4 \alpha^2 L^2  G^2 |\mathcal{N}_{g}|}{\left(\sum_{i \in \mathcal{N}_g}|\mathcal{N}_i|\right)^2} \sum_{i \in \mathcal{N}_g} |\mathcal{N}_{i}|^2 \frac{1}{t_i^u} \sum_{l=0}^{t_i^{u}-1}  l^2 .
\end{align}

\end{theorem}

\begin{remark}
The significance of controlling $S$ is reflected in the values of $\{t_{i}^{u}\}$ and $\mathcal{U}$: increasing $S$ allows more local iterations, $\{t_i^{u}\}$, between the LPSs and their associated clients, at the expense of performing less global communication rounds, $\mathcal{U}$, as a consequence of the system time $T$ being limited, and vice versa. The value of $S$ affects the convergence bound in (\ref{eq_global_bound}), since all the terms except the second one are functions of $t_{i}^{u}$ and $\mathcal{U}$.
\end{remark}
\begin{remark}
Since the terms including $\{t_{i}^{u}\}$ and $\mathcal{U}$ are in the form: $\sum_{u=1}^{\mathcal{U}}t_{i}^{u}$, this gives us an insight that one can tune $S$ to obtain similar performances under different individual values of $\{t_{i}^{u}\}$ and $\mathcal{U}$. This is discussed in more detail in our experimental results.
\end{remark}
\begin{remark}
Although we obtain a convergence bound that is a function of $S$, reflected in  $\{t_i^{u}\}$ and $\mathcal{U}$, its optimal value is not directly obtainable analytically from the bound. In our experiments, we show how different system scenarios lead to different choices of $S$ to optimize the performance.
\end{remark}
Now that we have seen how the sync time $S$ controls the upper bounds in the theorems above by statistically controlling the number of local iterations, and in order to characterize the convergence rate, let us assume that there exists a {\it minimum} local iteration time for group $i$:
\begin{align} \label{t_bound1}
\tau_{i,l}^u\geq c_i,~\text{a.s.},~\forall l,u.
\end{align}
Then, one gets a {\it maximum} number of local iterations:
\begin{align} \label{t_bound2}
t_{i}^{u} \leq t_{i}^{\max}\triangleq\ceil*{\frac{S}{c_i}},~\text{a.s.},~\forall u.
\end{align}
Based on this, one can get the following convergence rate (the proof is in Appendix~\ref{appE}):
\begin{corollary} [\textit{Speed of Convergence}]
\label{corollary}
For a given $\{t_i^{\max}\}$, setting $\alpha =\min\{\frac{1}{\sqrt{\mathcal{U}}},\frac{1}{L}\}$, the delay sensitive HFL algorithm achieves 
$   \frac{1}{\mathcal{U}} \sum_{u=1}^{\mathcal{U}} \mathbb{E}\left\|\nabla f\left(x^{u}\right)\right\|^2\leq\mathcal{O}(\frac{1}{\sqrt{\mathcal{U}}})$.
\end{corollary}
Therefore, for a finite sync time $S$, as the training time $T$ increases, the number of the global communication rounds $\mathcal{U}$ also increases, and hence Corollary~\ref{corollary} shows that the gradient converges to $0$ sublinearly.

\section{Experiments}\label{experiments}
In this section, we present our simulation results for the proposed delay sensitive HFL algorithm to verify the findings from the theoretical analysis.

\noindent \textbf{Datasets and Models.} We adopt a variety of data sets to demonstrate our findings and emphasize the effectiveness of our proposed algorithm. We consider an image classification supervised learning task on the CIFAR-10 dataset \cite{krizhevsky2009learning}, MNIST dataset and the Federated Extended MNIST dataset \cite{fedml_paper}. We also consider a word prediction task using the Shakespeare dataset \cite{fedml_paper}. %, consisting of 60000 32x32 color images in 10 classes, with 6000 images per class. There are 50000 training images and 10000 test images.
We employ different machine learning and deep learning models: (1) a convolution neural network (CNN) is adopted with two 5x5 convolution layers, two 2x2 max pooling layers,  two fully connected layers with 120 and 84 units, respectively, ReLu activation, a final softmax output layer and cross entropy loss; (2) a multi-layer perceptron (MLP) with one hidden layer consisting of 50 neurons, with a logistic regression model; and (3) for the Shakespeare dataset, we train an recurrent neural network (RNN) for next-character-prediction \cite{adaptive_opt}.

\noindent \textbf{Data Partitioning.} In order to validate our proposed HFL algorithm, we test its performance under different data partitioning methods. These are a baseline data partitioning method, where data is distributed in an i.i.d. manner across clients, and another method, where data is distributed in a non-i.i.d. manner with various levels of data heterogeneity. The data partitioning is conducted based on the Dirichlet distribution and it follows the same data partitioning as in \cite{fedml_paper}.

\noindent\textbf{Linear Delay Model.}  We consider shifted exponential delays \cite{shiftedexp1,delay_model}: $\tau_{i,l}^u\sim\exp(c_i,\lambda_i)$ and $\tau_g^u\sim\exp(c_g,\lambda_g)$. That is, $\mathbb{E}[\tau^u_{i,l}]=c_i+\frac{1}{\lambda_i}$ and $\mathbb{E}[\tau^u_g]=c_g+\frac{1}{\lambda_g}$.
To capture the relation between the number of clients and the number of groups with the delay, we  model the shift parameter for group $i$ as a linearly increasing function in the number of clients in that group: $c_i= d \times |\mathcal{N}_i| + b, \: \forall i \in \mathcal{N}_g$, for some constants $d$ and $b$. Similarly, we model the exponential rate as $1/\lambda_i=e \times |\mathcal{N}_i| + f, \: \forall i \in \mathcal{N}_g$. Thus, the intra-group delay increases as the number of its  clients increases. We adopt a similar model for the delay between LPSs and the GPS: $c_g= d_g \times |\mathcal{N}_g| + b_g $ and $1/\lambda_g=e_g \times |\mathcal{N}_g| + f_g$. Thus, the more LPSs a system has, the more intense the communication bottleneck becomes while communicating with the GPS. We will state the delay model values in this order: $\{d ,b ,e, f, d_g, b_g, e_g, f_g\}$.

\subsection{HFL Incentive and Motivation}

We first start with an HFL setting with $|\mathcal{N}_g|=2$ groups, with $10$ clients per group training the CNN model described above. In Fig.~\ref{1}, we show the evolution of both groups' accuracies and the global accuracy across time. The zoomed-in version in Fig.~\ref{fig:1b} shows the high (SGD) variance in the performance of the two groups especially during the earlier phase of training. Then, after more averaging with the GPS, the variance is reduced.
\begin{figure*}[htp] 
    \vspace{-.3in}
    \centering
    \subfloat[Performance over the whole training time.]{%
        \includegraphics[width=0.5\linewidth]{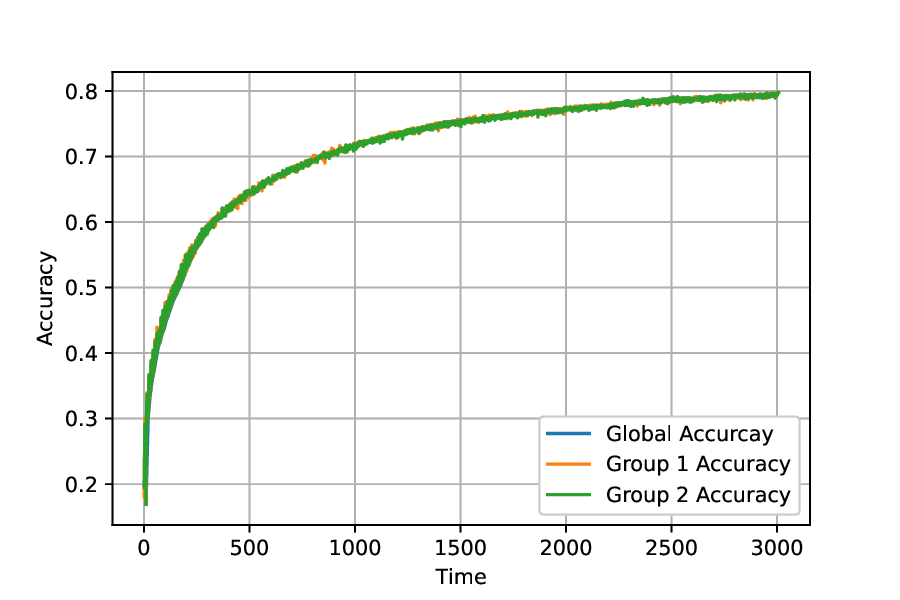}%
        \label{fig:1a}%
        }%
    \hfill%
    \subfloat[Performance at the beginning of training.]{%
        \includegraphics[width=0.5\linewidth]{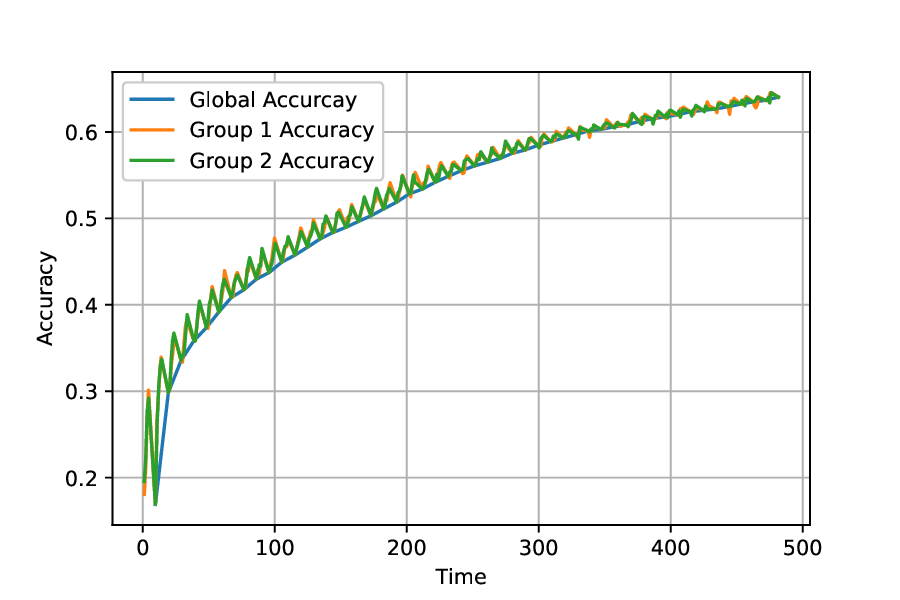}%
        \label{fig:1b}%
        }%
    \caption{HFL: 10 clients per group, parameters $\{0.09,0.1,0.009,0.01,1,3,0.05,0.1\}$, and $S=5$.}
    \label{1}
    
\end{figure*}

\begin{figure}[h]
\vspace{-.3in}
\centering
\includegraphics[width=0.5\linewidth]{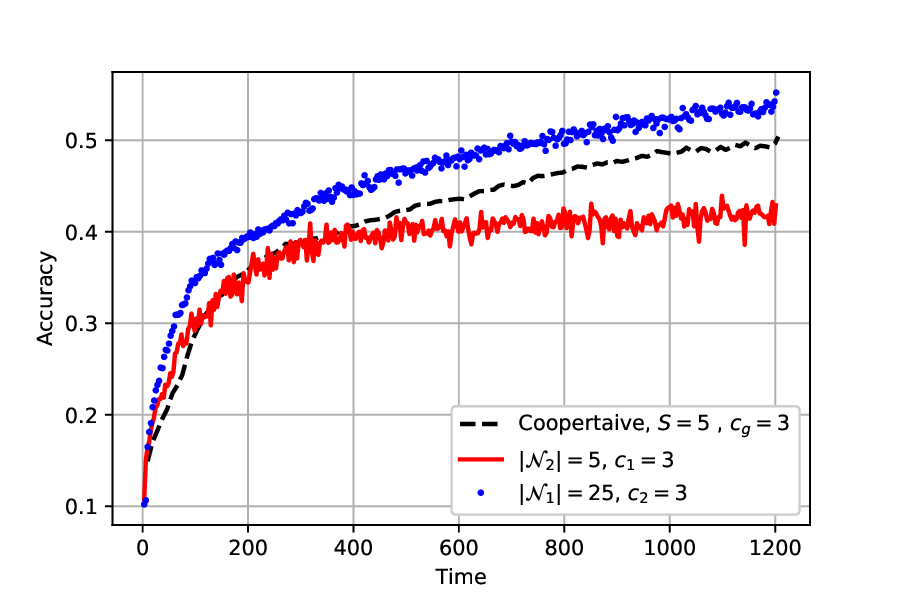}
\caption{Significance of group cooperation under non-i.i.d data.}
\label{3}
% \vspace{-.3in}
\end{figure}
In Fig.~\ref{3}, using the same CNN architecture on CIFAR-10, the significance of collaborative learning is emphasized. We change the number of clients per group, and run three experiments: one for each group in an isolated fashion, i.e., $S=T$, and one under the HFL setting ($S=5$). First, while we do not conduct our theoretical analysis under heterogeneous data distribution, we consider a \textcolor{black}{non-i.i.d.} data distribution among the two groups in this setting (Dirichlet(0.1)), and we see that our proposed algorithm still \textit{converges}. Second, it is clear that the performance of the group with fewer clients under heterogeneous data distribution and isolated learning will be deteriorated. However, aided by HFL, its performance improves, while the performance of the other group does not suffer greatly, promoting \textcolor{black}{\it fairness} among the groups.

To provide further insights on the significance of cooperation, we consider an extreme data heterogeneity setting in which each client only has two labels, with no common labels being shared among the two groups. In Fig.~\ref{4}, we plot our results and show that HFL \textit{converges} even when the groups are isolated (non-cooperative setting, i.e., $S = T$) yet since each group only has a strict subset of the total number of labels, they both converge to a lower accuracy compared to the one they get from cooperation. Further, in Fig.~\ref{fig:4b}, cooperation is shown to be significant in a non-extreme scenario in which the data distribution is Dirichlet(0.1). Different from Fig.~\ref{3}, the cooperative result in Fig.~\ref{fig:4b} beats both groups' isolated training results.

\begin{figure*}[h]
\vspace{-.3in}
\centering
\subfloat[No common labels between two groups.]{%
        \includegraphics[width=0.5\linewidth]{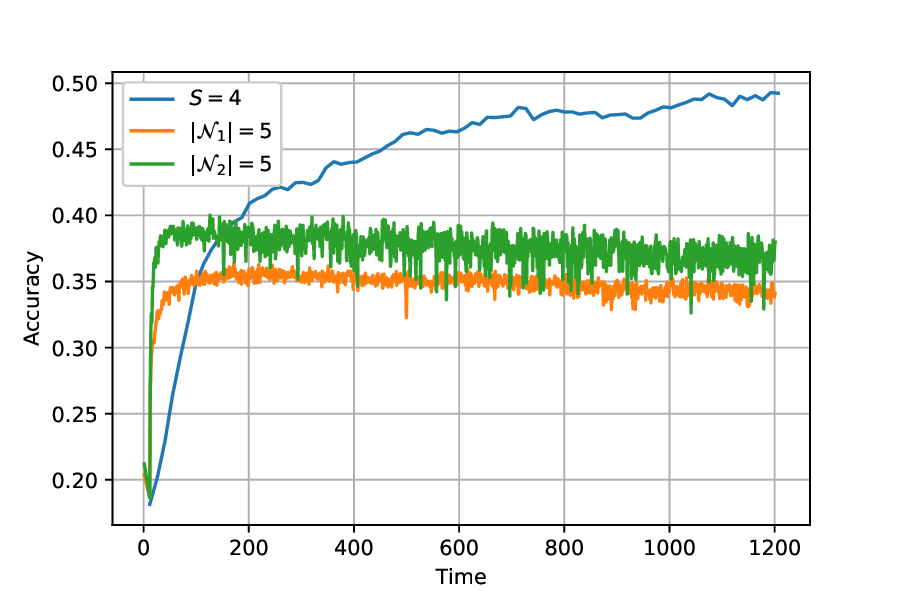}%
        \label{fig:4a}%
        }%
    \hfill%
    \subfloat[Data distribution follows Dirichlet(0.1).]{%
        \includegraphics[width=0.45\linewidth]{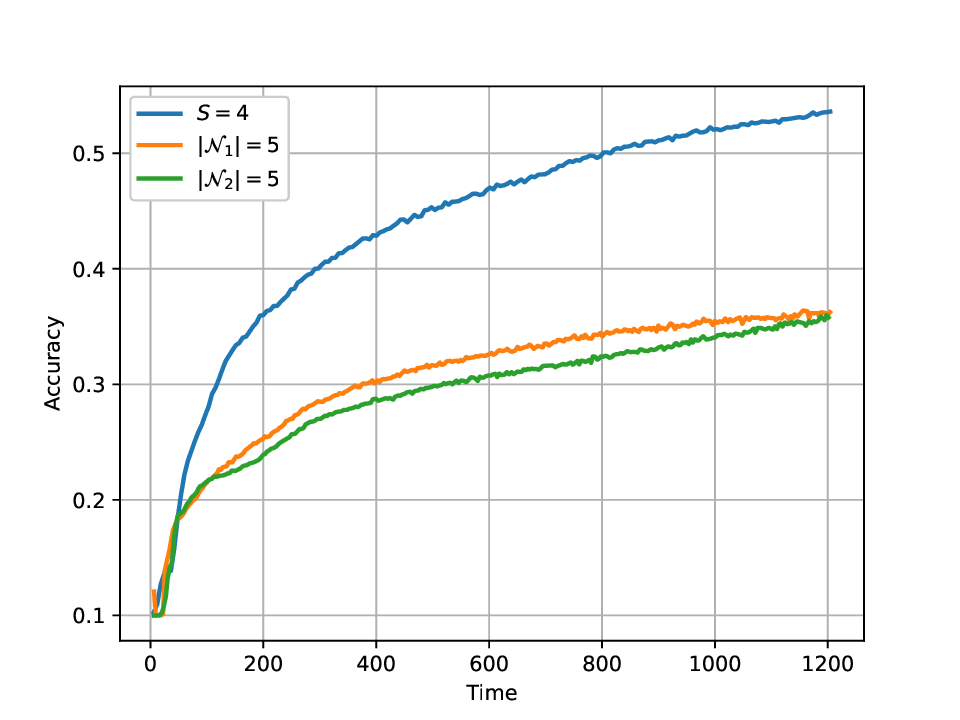}%
        \label{fig:4b}%
        }%
\caption{Significance of group cooperation.}
\label{4}

\end{figure*}

Hence, from the previous figures we can summarize that our proposed HFL setting is able to \textcolor{black}{enhance FL performance by reducing variance}, improving fairness, and learning more collaboratively by offering indirect access to datasets that are otherwise unavailable to clients.

\subsection{Significance of Choosing the Sync Time $S$}

In this section, we show how choosing $S$ can overcome the heterogeneity in the experienced delays among different groups, and can sustain a better performance for the whole system. Furthermore, we show that carefully choosing $S$ can outperform the baseline of setting $S=0$ (which corresponds to a centralized system (non-hierarchical)), with fewer visits to the GPS.

In Fig.~\ref{control_s_1}, we set the linear delay model parameters to $\{10^{-3},0.5,10^{-3},0.05,5,10, 10^{-3},0.05\}$, and train an MLP over the MNIST dataset. In this setting, we have two groups with 500 clients each. We can see that setting $S=5$ time units achieves the same performance as that of the baseline ($S=0$) yet with a {\it fewer number of global communication rounds}. The main reason behind this is that the cost of communicating with the GPS is relatively high in terms of the delay experienced in this setting.

\begin{figure}[h]
\vspace{-.3in}
\centering
\includegraphics[width=0.5\linewidth]{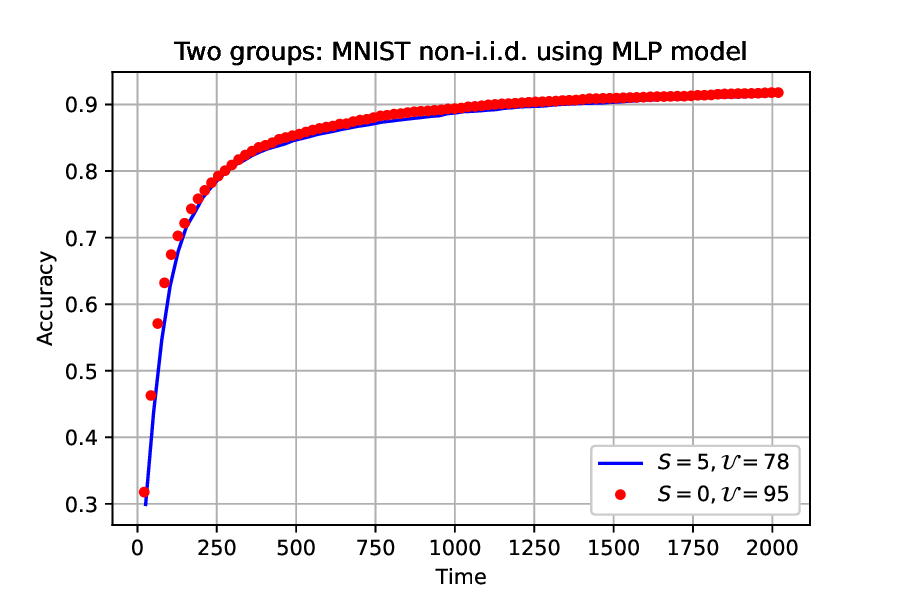}
\caption{Achieving the same performance as the baseline ($S=0$) with fewer global rounds.}
\label{control_s_1}

\end{figure}

Next, we show that we can {\it outperform} that baseline ($S=0$) by optimizing $S$. We train the CNN model on the CIFAR10 dataset for 30 clients divided among 2 groups, together with the challenge of having a non-i.i.d. data distribution. In Fig.~\ref{fig:5a}, we set the delay parameters to $\{0.01,0.85,10^{-3},0.085,4,2,0.4,0.2\}$. We see that $S=5$ beats the baseline with a notable gap and with fewer number of global rounds. In addition, we also see that in case there were no strict system time constraints, one may set $S=20$ which eventually approaches the baseline while saving $37\%$ in terms of global rounds. We next change $c_g$ to $30$ and rerun our experiments and show the results in Fig.~\ref{fig:5b}. We notice that $S=5$ is still the optimum choice, but in case the system has an additional constraint on communicating with the GPS, $S=20$ will be a better option, especially that the accuracy gain will not be sacrificed much. It is also worth noting that the system time $T$ (training time budget) plays a significant role in choosing $S$: $S=0$ (always communicate with the GPS) outperforms $S=20$ as long as $T \lessapprox 750$, and the opposite is true afterwards. This means that in some scenarios, the hierarchical setting may not be the optimal setting (which is different from the findings in \cite{Hfl_kh}). For instance, if the system has a stringent time constraint on learning, it would be better to communicate directly with GPS more frequently to get the advantage of learning the resulting models from different data.

\begin{figure*}[htp] 
    \vspace{-.3in}
    \centering
    \subfloat[$c_g=10$]{%
        \includegraphics[width=0.5\linewidth]{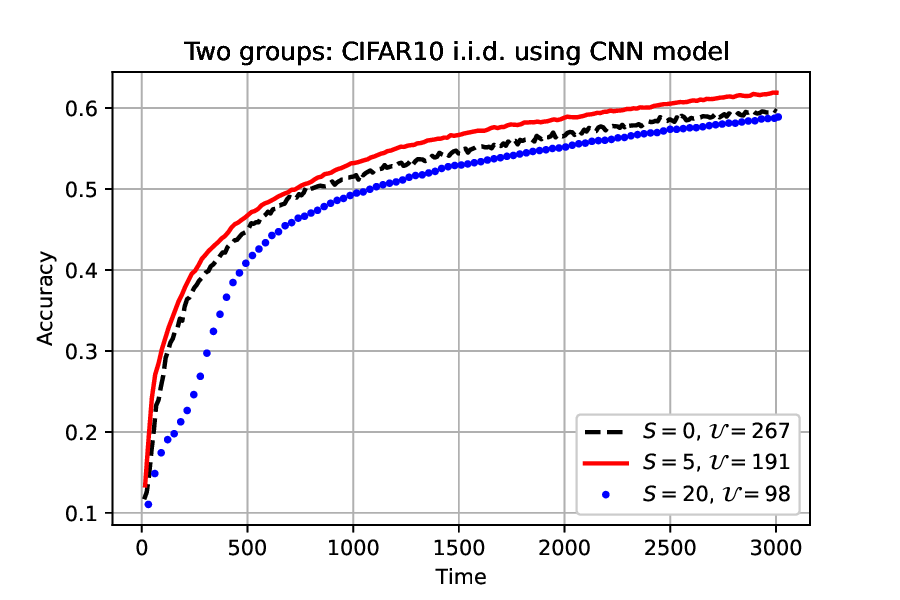}%
        \label{fig:5a}%
        }%
    \hfill%
    \subfloat[$c_g=30$]{%
        \includegraphics[width=0.5\linewidth]{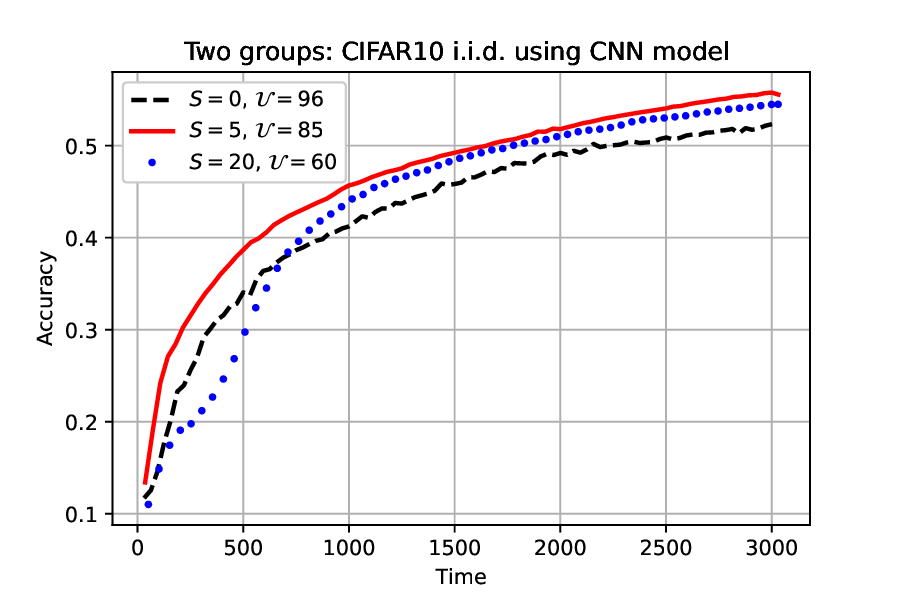}%
        \label{fig:5b}%
        }%
    \vspace{-.1in}
    \caption{Impact of the global shift parameter $c_g$ on choosing the sync time $S$.}
    \label{control_s_2}
   
\end{figure*}

\begin{figure}[h]

\centering
\includegraphics[width=0.5\linewidth]{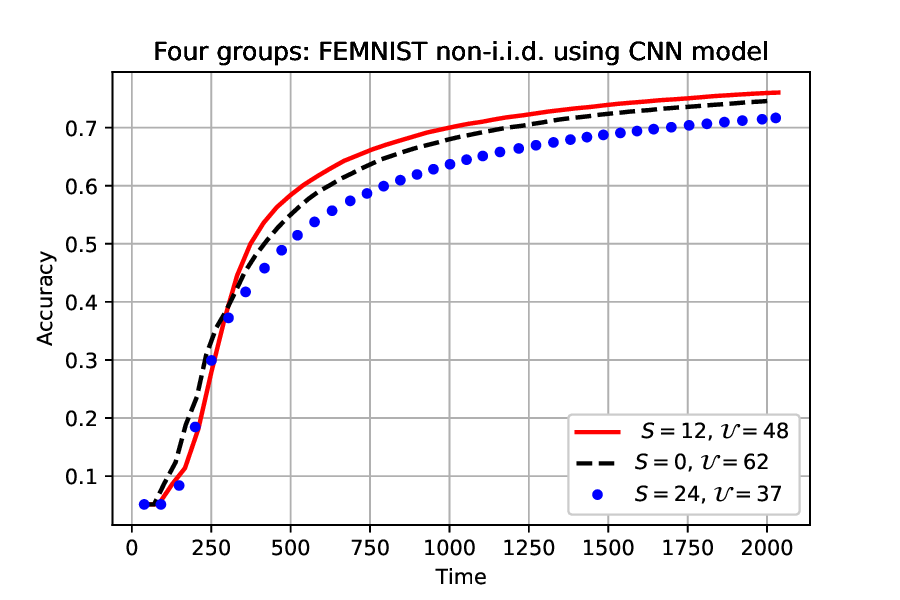}
\caption{Tuning $S$ to beat the baseline in a four group setting.}
\label{otp_s_w_4groups}

\end{figure}

To emphasize the significance of choosing $S$ in more challenging system setups, we train the CNN model on the Federated E-MNIST dataset under non-i.i.d. data distribution. We consider a large number of clients, 3400, divided equally among $|\mathcal{N}_g|=4$ groups, with delay parameters $\{ 0.1,  1.4, 0.01, 0.14, 4,  2, 0.5,  1\}$. In Fig.~\ref{otp_s_w_4groups}, the baseline $S=0$ beats $S=12$ and $S=24$ for $T\lessapprox300$, while eventually $S=12$ achieves the best performance. We note that one reason for the relatively larger choices of $S$ compared to previous cases is that the delay parameters are relatively larger. This emphasizes how $S$ depends on the intra- and inter-group delays.

\subsection{Clustering: Effect of the Number of Groups}

 \begin{figure}[h]
\centering
\includegraphics[width=0.5\linewidth]{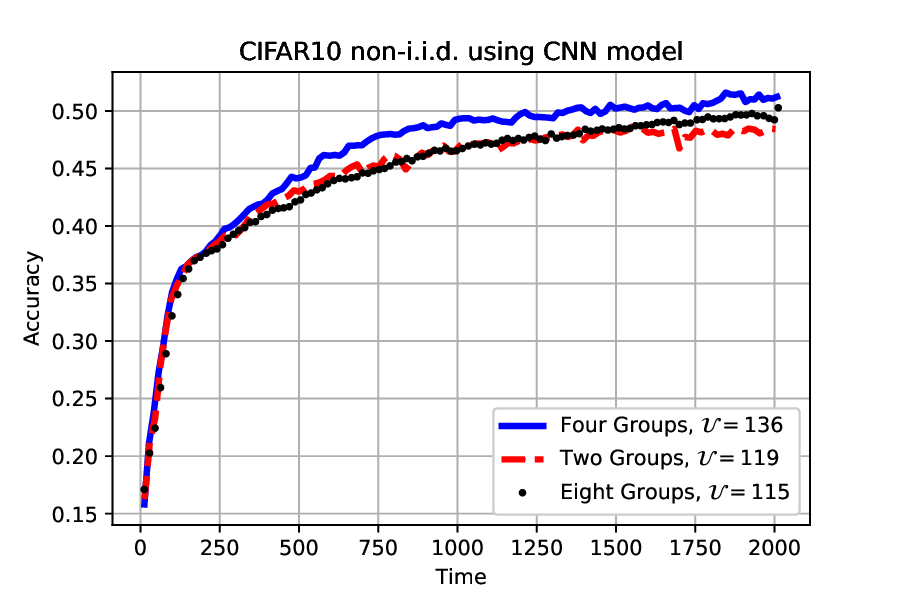}
\caption{Effect of the number of groups on the performance of HFL.}
\label{two_four_eight_fig}
% \vspace{-.3in}
\end{figure}
Adopting the linear delay model presents an interesting phenomenon that we investigate in this section. That is, increasing the number of groups results in increasing the delay between the GPS and the LPSs (i.e, $c_g$ and $1/\lambda_g$ will increase), but at the same time the delay between an LPS and its associated clients will decrease since there would be a smaller number of clients in a group (i.e, $c_i$ and $1/\lambda_i$ will decrease). In Fig.~\ref{two_four_eight_fig}, we consider a system of 64 clients divided equally between 2, 4, and 8 groups. The delay system parameters are given by $\{0.1,1.4,0.01,0.14,0.5,5,0.05,0.5\}$. We show that under the same choice of $S=6$, clustering clients into 4 groups is better than clustering them into 2 or 8 groups. We justify this since in the 4-group setting, the system strikes a balance between two extremes. That is, in the 2-group setting, the relatively high intra-group delay reduces the number of local iterations ($t_i^u$), while in the 8-group setting the relatively high LPS-to-GPS delay reduces the number of global rounds ($\mathcal{U}$).

\begin{figure*}[htp] 
% \vspace{-.3in}
    \centering
    \subfloat[$c_1=1 \text{ and} \:\: c_2=7$]{%
        \includegraphics[width=0.5\linewidth]{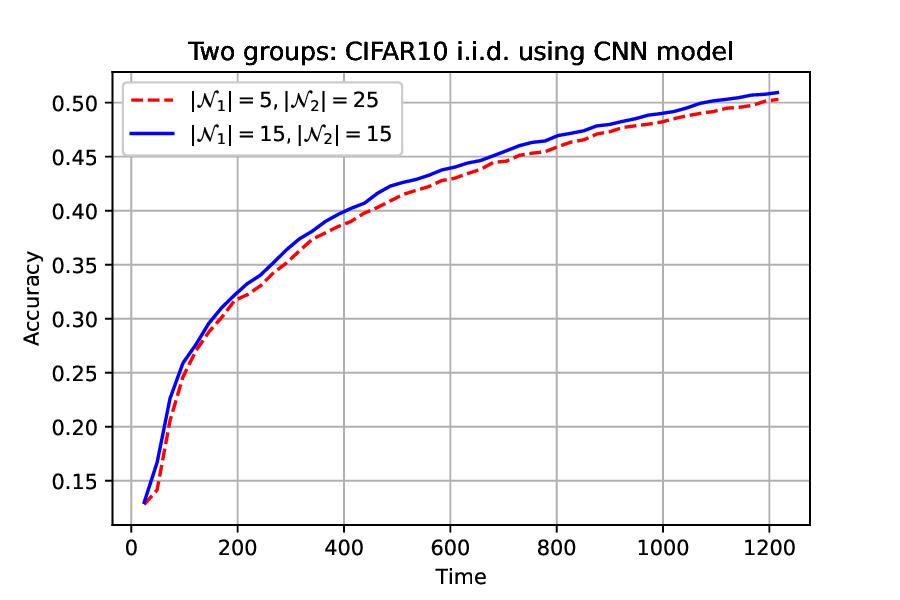}%
        \label{fig:UA_a}%
        }%
    \hfill%
    \subfloat[$c_1=7 \text{ and} \:\: c_2=1$]{%
        \includegraphics[width=0.5\linewidth]{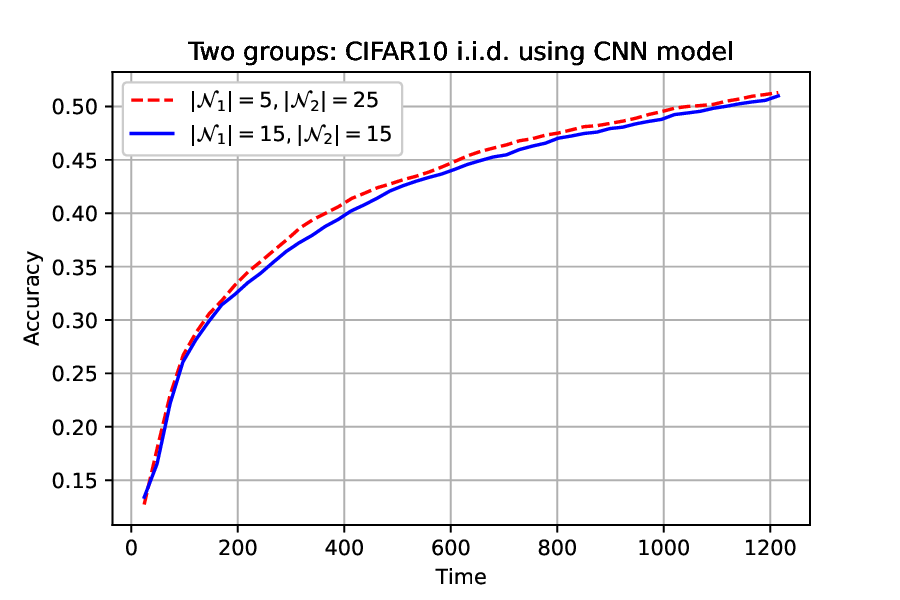}%
        \label{fig:UA_b}%
        }%
    \caption{Impact of the groups' shift parameters $c_1$ and $c_2$ on the group-client association under $S=8$ and $c_g=10$.}
    \label{userassociation}
% \vspace{-.3in}
\end{figure*}

\subsection{Effects of Client Association and Global Delay}

In this section, we study the effects of the shift parameters on the overall performance of HFL, with a focus on client association and LPS-to-GPS delay.

In Fig.~\ref{userassociation}, the effect of the groups' shift parameters $c_1$ and $c_2$ on determining the optimal group-client association is investigated. The results show that it is not always optimal to cluster the clients evenly among the groups. In Fig.~\ref{fig:UA_b} for instance, we see that assigning less clients to a group with a relatively smaller shift parameter performs better than an equal assignment of clients among both groups; this is observation is reversed in Fig.~\ref{fig:UA_a}, in which a larger number of clients is assigned to the relatively slower LPS.

%%%%%%%%%%%%%%%%%%%
\begin{figure}[h]
\centering
\includegraphics[width=0.5\linewidth]{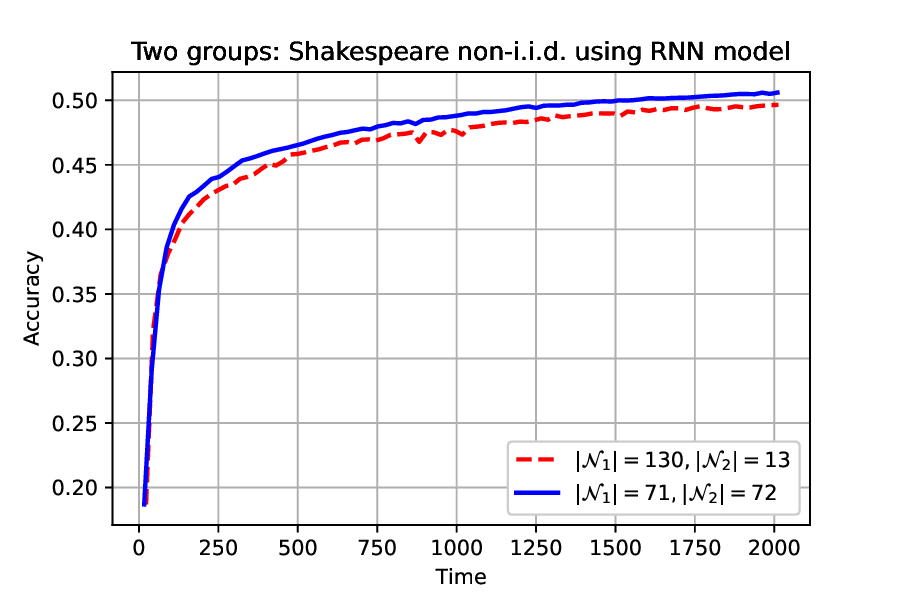}
\caption{The symmetric and asymmetric user association under linear delay model.}
\label{rnn_shk}
% \vspace{-.3in}
\end{figure} 

In Fig.~\ref{rnn_shk}, we train the RNN model on the Shakespeare dataset under non-i.i.d. data distribution with $S=12$. The delay parameters are set to $\{0.1,0.7,0.001,0.13,3,3,0.01,0.2\}$.  
We observe that the symmetric setting outperforms the non-symmetric one. The main reason behind this is that in the asymmetric case, the second group with fewer clients has to idly wait for the other group until it finishes its iterations, and hence the straggler effect becomes dominant.

%%%%%%%%%%%%%%%%%%%%%
\begin{figure}[h]
\centering
\includegraphics[width=0.45\linewidth]{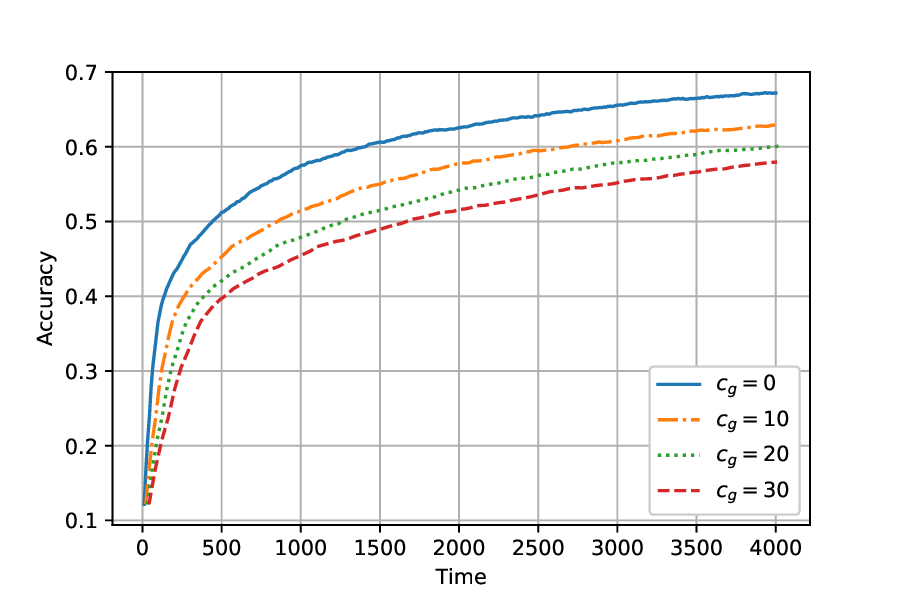}
\caption{The effect of global shift parameter $c_g$ under $S=10$.}
\label{5}

\end{figure}

Finally, in Fig.~\ref{5}, the impact of global shift parameter $c_g$ on the global accuracy is shown. As the global shift delay parameter increases, the performance  gets worse. This is mainly because the number of global communication rounds with the GPS, $\mathcal{U}$, is reduced, which hinders the clients from getting the benefit of accessing other clients' learning models. 

\noindent\textbf{Future Investigation.} Guided by our understanding from the convergence bounds and the simulation results, we observe that it is better to make the parameter $S$ \textit{variable} especially during the first global communication rounds. For instance, instead of fixing $S=5$, we allow $S$ to increase gradually with each round from $1$ to $5$, and then fix it at $5$ for the remaining rounds. Our reasoning behind this is that the clients' models need to be {\it directed} towards global optimum, and not their local optima. Since this direction is done through the GPS, it is reasonable to communicate with it more frequently at the beginning of learning to push the local models towards the optimum direction. To investigate this setting, we train a logistic regression model over the MNIST dataset, and distribute it in a non-iid fashion over 500 clients per group. As shown in Fig.~\ref{svariable}, the variable $S$ approach achieves a higher accuracy than the fixed one, with the effect more pronounced as $S$ increases.

\begin{figure*}[htp] 
    \vspace{-.3in}
    \centering
    \subfloat[$S=5$]{%
        \includegraphics[width=0.33\linewidth]{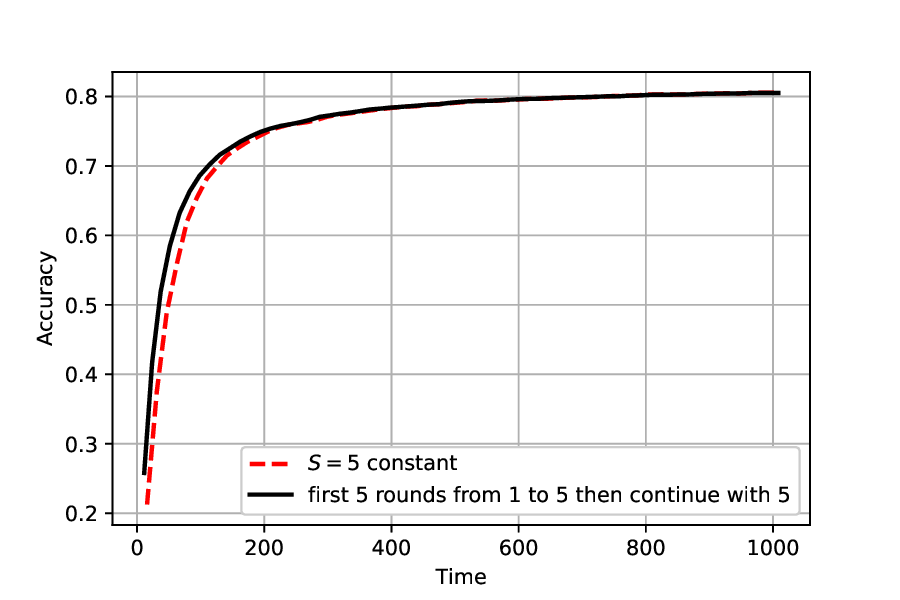}%
        \label{fig:a}%
        }%
    \hfill%
    \subfloat[$S=10$]{%
        \includegraphics[width=0.33\linewidth]{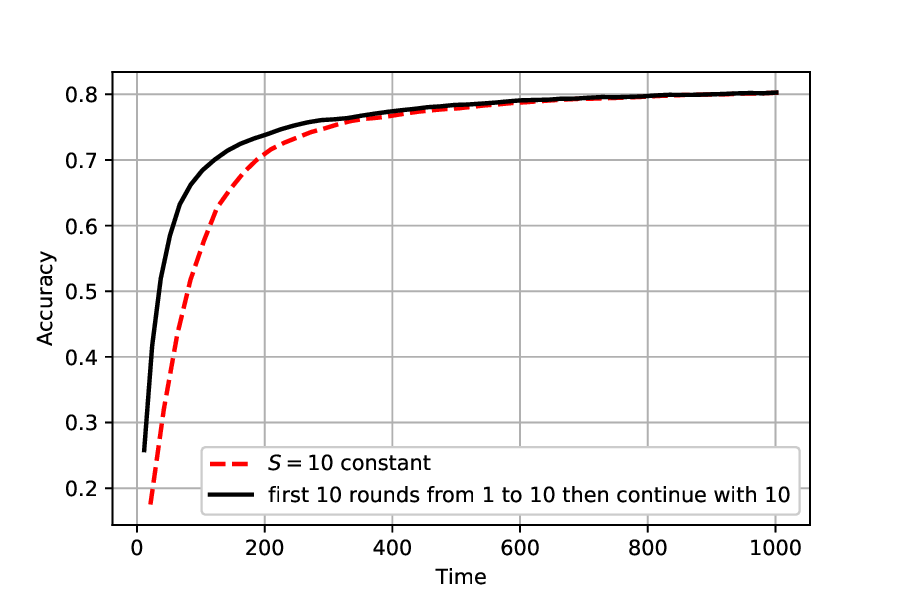}%
        \label{fig:b}%
        }%
        \hfill%
    \subfloat[$S=20$]{\includegraphics[width=0.33\linewidth]{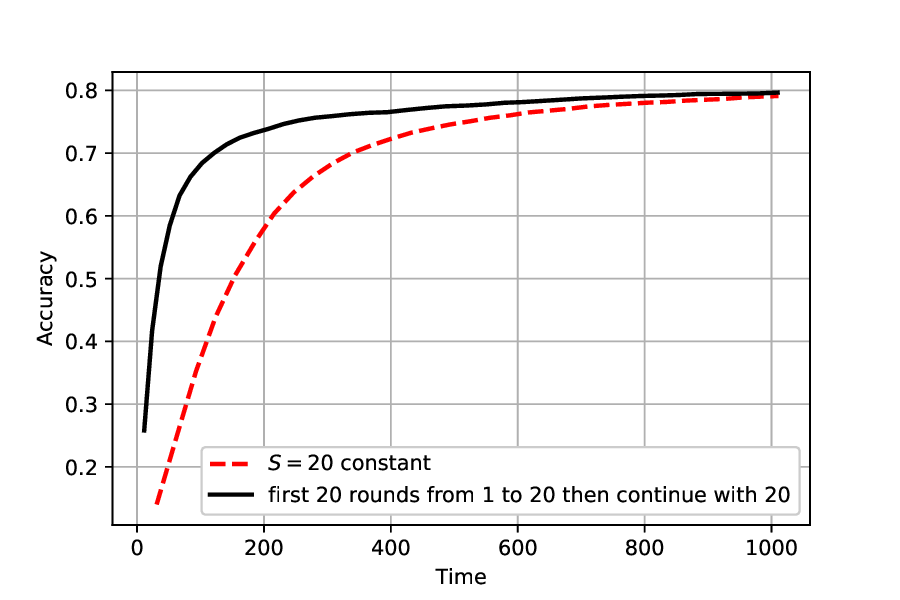}%
        \label{fig:c}%
        }
    \caption{Comparison between variable and fixed $S$ with respect to the global learning accuracy.}
    \label{svariable}

\end{figure*}
%===
\section{Conclusion}\label{conclusion}
A delay sensitive HFL algorithm has been proposed, in which the effects of wall-clock times and delays on the accuracy of FL \textcolor{black}{are investigated}. A sync time $S$ governs how many local iterations are allowed at LPSs before forwarding to the GPS, and a system time $T$ constrains the overall training period. Our findings reveal that the optimal $S$ depends on different factors such as the delays at the LPSs and the GPS, the number of clients per group, and the value of $T$. Multiple insights are drawn on the performance of HFL in time-restricted training settings. 

Mitigating the impact of stragglers in delay sensitive HFL could be further mitigated via effective scheduling, or through changing the system into asynchronous learning. Both directions will be investigated in future work.
\appendices

\section{Preliminaries}
We will use the following relationships throughout our proofs without explicit reference:

\begin{itemize}
\item For any $x, y \in \mathbb{R}^n$, we have: $\langle  x,y\rangle \leq  \frac{1}{2}\left\|x\right\|^2+ \frac{1}{2}\left\|y\right\|^2$.   
\item By Jensen's inequality, for $x_{i} \in \mathbb{R}^n$, $i \in \{1,2,3,\dots,N\}$, we have
$\left\|\frac{1}{N}\sum_{i=1}^{N}x_i\right\|^2 \leq \frac{1}{N}\sum_{i=1}^{N}\left\|x_i\right\|^2$,
which implies $\left\|\sum_{i=1}^{N}x_i\right\|^2 \leq N\sum_{i=1}^{N}\left\|x_i\right\|^2$.
\end{itemize}

\section{Proof of Lemma~\ref{lemma_1}}\label{appB}
Conditioning on the number of local updates of group $i$ up to and including global round $u$, ${\bm t}_i^{u}$, we evaluate the expected difference between the aggregated global model and the latest local model at group $i$, by the end of global round $u$. Based on \eqref{eq_local-update} and \eqref{global_update}, the following holds:
\begin{align}
\mathbb{E}&_{|\bm t_i^u}\left\| x^{u+1}-x_{i}^{u,t_{i}^{u}}\right\|^2 =\alpha^2 \mathbb{E}_{|\bm t_i^u}\!\left\| \frac{1}{|\mathcal{N}_i|} \underset{{k \in \mathcal{N}_i}}{\sum}\!\overset{t_{i}^{u}-1}{\underset{l=0}\sum}\Tilde{g}_{i,k}\left(x_{i}^{u,l}\right)- \frac{1}{ \underset{i \in \mathcal{N}_g}{\sum}|\mathcal{N}_i|} \underset{i \in \mathcal{N}_g}{\sum}\! \frac{1}{t_{i}^{u}}\underset{k \in \mathcal{N}_i}{\sum} \overset{t_{i}^{u}-1}{\underset{l=0}\sum}\Tilde{g}_{i,k}\left(x_{i}^{u,l}\right)\right\|^2 \!\!\!  \nonumber\\
\leq& 2 \alpha^2  \mathbb{E}_{|\bm t_i^u} \frac{1}{|\mathcal{N}_{i}|^2} \left\| \sum_{k \in \mathcal{N}_i} \sum_{l=0}^{t_{i}^{u}-1}\Tilde{g}_{i,k}\left(x_{i}^{u,l}\right)\right\|^2 + 2 \alpha^2 \mathbb{E}_{|\bm t_i^u} \frac{|\mathcal{N}_g|}{\left(\sum_{i \in \mathcal{N}_g}|\mathcal{N}_i|\right)^2}\!\sum_{i \in \mathcal{N}_g}\! \!\frac{1}{(t_{i}^{u})^2}\left\| \sum_{k \in \mathcal{N}_i}\! \sum_{l=0}^{t_{i}^{u}-1}\Tilde{g}_{i,k}\left(x_{i}^{u,l}\right)\right\|^2 \!\!  \nonumber\\
\leq& 2 \alpha^2 \left( \frac{1}{|\mathcal{N}_{i}|^2} + \frac{|\mathcal{N}_g|}{\left( \sum{i \in \mathcal{N}_g}|\mathcal{N}_i|\right)^2(t_i^{u})^2}\right)  \mathbb{E}_{|\bm t_i^u} \left\| \sum_{k \in \mathcal{N}_i } \sum_{l=0}^{t_{i}^{u}-1}\Tilde{g}_{i,k}\left(x_{i}^{u,l}\right)\right\|^2 \nonumber \\
&+\frac{2 \alpha^2 |\mathcal{N}_g|}{\left(\sum_{i \in \mathcal{N}_g}|\mathcal{N}_i|\right)^2}\sum_{j \in \mathcal{N}_g \setminus{\{i\}}}\frac{1}{(t_j^{u})^2}\mathbb{E}_{|\bm t_i^u} \left\|\sum_{k \in \mathcal{N}_j} \sum_{l=0}^{t_{j}^{u}-1}\Tilde{g}_{j,k}\left(x_{j}^{u,l}\right)\right\|^2 \nonumber\\
\leq&  \left(\frac{2 \alpha^2 }{|\mathcal{N}_{i}|} + \frac{2 \alpha^2 |\mathcal{N}_g||\mathcal{N}_i|}{\left(\sum_{i \in \mathcal{N}_g}|\mathcal{N}_i|\right)^2(t_i^u)^2}\right) \sum_{k \in \mathcal{N}_i} \mathbb{E}_{|\bm t_i^u}\left\|\sum_{l=0}^{t_{i}^{u}-1}\!\!\Tilde{g}_{i,k}\left(x_{i}^{u,l}\right)\right\|^2 \nonumber \\
&+\frac{2 \alpha^2|\mathcal{N}_g|}{\left( \sum_{i \in \mathcal{N}_g}|\mathcal{N}_i|\right)^2}\sum_{j \in \mathcal{N}_g \setminus{\{i\}}} \frac{|\mathcal{N}_j|}{(t_j^u)^2}
\sum_{k \in \mathcal{N}_j}\!\!\mathbb{E}_{|\bm t_i^u}\left\| \sum_{l=0}^{t_{j}^{u}-1}\Tilde{g}_{j,k}\left(x_{j}^{u,l}\right)\right\|^2. 
\end{align}
The proof now follows by applying Assumption~3 (bounded gradient in \eqref{assum_3}) to the bound above.
\vspace{-.5in}
\section{Proof of Theorem~\ref{CA_Group}}\label{appC}
\vspace{-.1in}
Based on the smoothness assumption of the loss function at LPS $i$, the SGD update rule in \eqref{eq_lps-model-itr}, and the local aggregation rule in \eqref{eq_lps-model-agg}, one can write
\begin{align} \label{smooth_bound}
\mathbb{E}_{|\bm t_i^u}f_{i}\left(x_{i}^{u,l}\right) \leq& \mathbb{E}_{|\bm t_i^u}f_{i}\left(x_{i}^{u,l-1}\right)+\frac{L}{2} \mathbb{E}_{|\bm t_i^u}\left\|x_{i}^{u,l}-x_{i}^{u,l-1}\right\|^2
+\mathbb{E}_{|\bm t_i^u} \langle\nabla f_{i}\left(x_{i}^{u,l-1}\right),x_{i}^{u,l}-x_{i}^{u,l-1}\rangle 
\nonumber\\
=&\mathbb{E}_{|\bm t_i^u}f_{i}\left(x_{i}^{u,l-1}\right) 
+   \frac{\alpha^2 L}{2}  \mathbb{E}_{|\bm t_i^u}\left\| \sum_{k \in \mathcal{N}_i} \frac{1}{|\mathcal{N}_i|} \Tilde{g}_{i,k}\left(x_{i}^{u,l-1}\right)\right\|^2 \nonumber \\
&+ \alpha \mathbb{E}_{|\bm t_i^u} \langle \nabla f_{i}\left(x_{i}^{u,l-1}\right),  \frac{-1}{|\mathcal{N}_{i}|} \sum_{k \in \mathcal{N}_i}\Tilde{g}_{i,k}\left(x_{i}^{u,l-1}\right) \rangle.
\end{align}
For the inner product term above, we have
\begin{align} \label{dot_bound}
&2 \mathbb{E}_{|\bm t_i^u} \langle \nabla f_{i}\left(x_{i}^{u,l-1}\right), \frac{-1}{|\mathcal{N}_i|} \sum_{k \in \mathcal{N}_i} \Tilde{g}_{i,k}\left(x_{i}^{u,l-1}\right)\rangle \overset{(\text{i})}{=} \mathbb{E}_{|\bm t_i^u} \left\|\nabla f_{i}\left(x_{i}^{u,l-1}\right)-  \frac{1}{|\mathcal{N}_i|}\!\!\sum_{k \in \mathcal{N}_i} \!\!\nabla F_{i,k}\left(\!x_{i}^{u,l-1}\!\right)\right\|^2\!\!\nonumber \\
&\hspace{2.25in}-\mathbb{E}_{|\bm t_i^u}\!\!\left\|\nabla f_{i}\!\left(\!x_{i}^{u,l-1}\!\right)\!\right\|^2\!\!\!\!-\!\mathbb{E}_{|\bm t_i^u}\left\|\frac{1}{|\mathcal{N}_i|}\sum_{k \in \mathcal{N}_i}\!\! \nabla F_{i,k}\!\!\left(x_{i}^{u,l-1}\right)\right\|^2,
\end{align}
where (i) follows from Assumption~2 (unbiased stochastic gradient in \eqref{assum_2}), and then applying  $ \langle x,y\rangle\!\!=\!\!\frac{1}{2}\left(\left\|x+y\right\|^2-\left\|x\right\|^2-\left\|y\right\|^2 \right)$. Regarding second term in \eqref{smooth_bound}, the following holds: 
\begin{align} \label{sgd_bound}
\mathbb{E}_{|\bm t_i^u}\left\| \frac{1}{|\mathcal{N}_i|} \sum_{k \in \mathcal{N}_i} \Tilde{g}_{i,k}\left(x_{i}^{u,l-1}\right)\right\|^2
 &\overset{(\text{ii})}{=}\frac{1}{|\mathcal{N}_i|^2} \sum_{k \in \mathcal{N}_i} \mathbb{E}_{|\bm t_i^u}\left\| \Tilde{g}_{i,k}\left(x_{i}^{u,l-1}\right)-\nabla F_{i,k}\left(x_{i}^{u,l-1}\right)\right\|^2\nonumber \\
 &\hspace{2in}+ \mathbb{E}_{|\bm t_i^u}\left\| \frac{1}{|\mathcal{N}_i|} \sum_{k \in \mathcal{N}_i} \nabla F_{i,k}\left(x_{i}^{u,l-1}\right)\right\|^2  \nonumber \\
 &\leq \frac{\sigma^2}{|\mathcal{N}_i|} +\mathbb{E}_{|\bm t_i^u}\left\| \frac{1}{|\mathcal{N}_i|}\sum_{k \in \mathcal{N}_i} \nabla F_{i,k}\left(x_{i}^{u,l-1}\right)\right\|^2,
\end{align}
where (ii) follows from adding and subtracting $\nabla F_{i,k}\left(x_{i}^{u,l-1}\right)$, and that each $k$th term $\Tilde{g}_{i,k}\left(x_{i}^{u,l-1}\right)-\nabla F_{i,k}\left(x_{i}^{u,l-1}\right)$ has zero mean and the overall $|\mathcal{N}_i|$ terms are independent across different clients. Substituting \eqref{dot_bound} and \eqref{sgd_bound} into \eqref{smooth_bound}:
\begin{align} 
\mathbb{E}&_{|\bm t_i^u}f_{i}\left(x_{i}^{u,l}\right) \leq \mathbb{E}_{|\bm t_i^u} f_{i}\left(x_{i}^{u,l-1}\right) + \frac{\alpha^2L\sigma^2}{2|\mathcal{N}_i|}+ \frac{\alpha^2 L}{2} \mathbb{E}_{|\bm t_i^u}\left\| \frac{1}{|\mathcal{N}_i|}\sum_{k \in \mathcal{N}_i} \nabla F_{i,k}\left(x_{i}^{u,l-1}\right)\right\|^2 \nonumber \\
&+ \frac{\alpha}{2}\Biggl(\mathbb{E}_{|\bm t_i^u} \left\|\nabla f_{i}\left(x_{i}^{u,l-1}\right)-  \frac{1}{|\mathcal{N}_i|}\sum_{k \in \mathcal{N}_i} \nabla F_{i,k}\left(x_{i}^{u,l-1}\right)\right\|^2 -\mathbb{E}_{|\bm t_i^u}\left\|\nabla f_{i}\left(x_{i}^{u,l-1}\right)\right\|^2\!\!\!\!\nonumber \\
&\hspace{3in}-\mathbb{E}_{|\bm t_i^u}\left\|\frac{1}{|\mathcal{N}_i|}\!\!\sum_{k \in \mathcal{N}_i} \nabla F_{i,k}\left(x_{i}^{u,l-1}\right)\right\|^2\!\Biggl) \nonumber \nonumber \\
=&\mathbb{E}_{|\bm t_i^u}f_{i}\left(x_{i}^{u,l-1}\right) - \frac{\alpha}{2} \mathbb{E}_{|\bm t_i^u}\left\|\nabla f_{i}\left(x_{i}^{u,l-1}\right)\right\|^2  +  \frac{\alpha^2 L \sigma^2 }{2|\mathcal{N}_i|} \nonumber \\
&\hspace{2in}-\left(\frac{\alpha}{2}- \frac{\alpha^2 L}{2}\right)\mathbb{E}_{|\bm t_i^u}\left\| \frac{1}{|\mathcal{N}_i|}\sum_{k \in \mathcal{N}_i} \nabla F_{i,k}\left(x_{i}^{u,l-1}\right)\right\|^2 \label{eq_nameless} \\
\leq&\mathbb{E}_{|\bm t_i^u}f_{i}\left(x_{i}^{u,l-1}\right) - \frac{\alpha}{2} \mathbb{E}_{|\bm t_i^u}\left\|\nabla f_{i}\left(x_{i}^{u,l-1}\right)\right\|^2+ \frac{\alpha^2 L \sigma^2 }{2|\mathcal{N}_i|} ,
\end{align}
%%%%%%%%%%%%%%%%%%%%%%%%%%%
where \eqref{eq_nameless} follows from \eqref{group_loss}, and the last inequality follows by choosing $0 < \alpha \leq \frac{1}{L}$.
Next, rearranging the terms above and summing over all local iterations until $t_{i}^{u}$, we have
\begin{align}
      \frac{\alpha}{2} \sum_{l=1}^{t_{i}^{u}}\mathbb{E}_{|\bm t_i^u}\left\|\nabla f_{i}\left(x_{i}^{u,l-1}\right)\right\|^2 &\leq \sum_{l=1}^{t_{i}^{u}} \left[\mathbb{E}_{|\bm t_i^u}f_{i}\left(x_{i}^{u,l-1}\right) - \mathbb{E}_{|\bm t_i^u}f_{i}\left(x_{i}^{u,l}\right) \right] + t_{i}^{u}  \frac{\alpha^2 L\sigma^2 }{2|\mathcal{N}_i|} \nonumber \\
    & = \left[\mathbb{E}_{|\bm t_i^u}f_{i}\left(x_{i}^{u}\right) - \mathbb{E}_{|\bm t_i^u}f_{i}\left(x_{i}^{u,t_{i}^{u}}\right) \right] +  t_{i}^{u} \frac{\alpha^2L\sigma^2}{2|\mathcal{N}_i|}.
\end{align}
Now taking the average over all global communication rounds yields
\begin{align} \label{global_avg}
 &\frac{\alpha}{2\sum_{u=1}^{\mathcal{U}} t_{i}^{u}} \sum_{u=1}^{\mathcal{U}}  \sum_{l=1}^{t_{i}^{u}} \mathbb{E}_{|\bm t_i^u}\left\|\nabla f_{i}\left(x_{i}^{u,l-1}\right)\right\|^2 
       \nonumber \\
&\leq   \frac{\alpha^2 L\sigma^2}{2|\mathcal{N}_i|}   +\frac{1}{\sum_{u=1}^{\mathcal{U}} t_{i}^{u}} \!\left(\!\mathbb{E}_{|\bm t_i^u}f_{i}\left(x_{i}^{1}\right)-\mathbb{E}_{|\bm t_i^u}f_{i}\left(x_{i}^{\mathcal{U},t_i^{\mathcal{U}}}\right) \!+\!\sum_{u=1}^{\mathcal{U}-1} 
    \mathbb{E}_{|\bm t_i^u}f_{i}\left(x_{i}^{u+1}\right)\!-\!\mathbb{E}_{|\bm t_i^u}f_{i}\left(x_{i}^{u,t_i^{u}}\right)\!\!\right).
       \end{align}
Now let us consider one of the summands in the equality above. We have
\begin{align}\label{last_step}
     \mathbb{E}_{|\bm t_i^u}f_{i}&\left(x_{i}^{u+1}\right)-\mathbb{E}_{|\bm t_i^u}f_{i}\left(x_{i}^{u,t_i^{u}}\right) \leq \mathbb{E}_{|\bm t_i^u}\langle\nabla f_{i}\left(x_{i}^{u,t_i^{u}}\right) ,x_{i}^{u+1}- x_{i}^{u,t_i^{u}}\rangle+ \frac{L}{2}\mathbb{E}_{|\bm t_i^u}\left\| x_{i}^{u+1}-x_{i}^{u,t_{i}^{u}}\right\|^2 \nonumber \\
  &\leq {\frac{1}{2}} \mathbb{E}_{|\bm t_i^u} \left(\left\|\nabla f_{i}\left(x_{i}^{u,t_i^{u}}\right)\right\|^2+ \left\|x_{i}^{u+1}- x_{i}^{u,t_i^{u}}\right\|^2\right) +\frac{L}{2}\mathbb{E}_{|\bm t_i^u}\left\| x_{i}^{u+1}-x_{i}^{u,t_{i}^{u}}\right\|^2 \nonumber \\
   &\leq  \frac{G^2}{2}+ (L +1)\alpha^2 \left( ( t_{i}^{u})^2+ \frac{|\mathcal{N}_g|}{\left(\sum_{i \in \mathcal{N}_g}|\mathcal{N}_i|\right)^2} \sum_{j \in \mathcal{N}_g }|\mathcal{N}_j|^2\right) G^2,
\end{align}
where the last inequality follows directly from Lemma~\ref{lemma_1} (note that each group restarts its model updates following each global iteration, and hence $x_{i}^{u+1,0}=x^{u+1}$). Finally, by substituting \eqref{last_step} into \eqref{global_avg} we get
\begin{align} 
&\frac{1}{\sum_{u=1}^{\mathcal{U}} t_{i}^{u}} \sum_{u=1}^{\mathcal{U}}  \sum_{l=1}^{t_{i}^{u}} \mathbb{E}_{|\bm t_i^u}\left\|\nabla f_{i}\left(x_{i}^{u,l-1}\right)\right\|^2  \leq \frac{2}{\alpha \sum_{u=1}^{\mathcal{U}} t_{i}^{u}}  
    \Biggl(\mathbb{E}_{|\bm t_i^u}f_{i}\left(x_{i}^{1}\right)-\mathbb{E}_{|\bm t_i^u}f_{i}\left(x_{i}^{\mathcal{U},t_i^{\mathcal{U}}}\right) \Biggl)+ \frac{ \alpha  L \sigma^2}{|\mathcal{N}_{i}|} \nonumber \\
    &\hspace{1.25in}+\frac{(\mathcal{U}-1) G^2 }{\alpha \sum_{u=1}^{\mathcal{U}} t_{i}^{u}}  + \frac{2 G^2(L +1)\alpha}{ \sum_{u=1}^{\mathcal{U}} t_{i}^{u}}   \sum_{u=1}^{\mathcal{U}-1}   \Bigg( ( t_{i}^{u})^2+ \underbrace{\frac{|\mathcal{N}_g|}{\left(\sum_{i \in \mathcal{N}_g}|\mathcal{N}_i|\right)^2} \sum_{j \in \mathcal{N}_g }|\mathcal{N}_j|^2}_{\triangleq\kappa}\Bigg)\!\!
    \nonumber\\
    &= \frac{2}{\alpha \sum_{u=1}^{\mathcal{U}} t_{i}^{u}} \Biggl(\mathbb{E}_{|\bm t_i^u}f_{i}\left(x_{i}^{1}\right)-\mathbb{E}_{|\bm t_i^u}f_{i}\left(x_{i}^{\mathcal{U},t_i^{\mathcal{U}}}\right) \Biggl)+ \frac{\alpha  L \sigma^2 }{|\mathcal{N}_{i}|} +\frac{2 (L +1)\alpha}{ \sum_{u=1}^{\mathcal{U}} t_{i}^{u}}   \sum_{u=1}^{\mathcal{U}-1}   ( t_{i}^{u})^2 G^2 \nonumber \\
    &\hspace{3in}+\Biggl(\frac{1}{\alpha \sum_{u=1}^{\mathcal{U}} t_{i}^{u}} +  \frac{2 (L +1) \kappa \alpha}{ \sum_{u=1}^{\mathcal{U}} t_{i}^{u}} \Biggr)(\mathcal{U}-1)  G^2. \end{align}
\section{Proof of Theorem~\ref{global_convg}}\label{appD}
\vspace{-.2in}
We first use the smoothness assumption of the global loss function, together with the SGD update rule in \eqref{global_update} to get the following:
\begin{align} \label{global_smooth_bound}
  \hspace{-3in}  \mathbb{E}_{|\bm t_i^u}f\left(x^{u+1}\right) \leq& \mathbb{E}_{|\bm t_i^u} f\left(x^{u}\right)+ \frac{\alpha^2 L}{2(\sum_{i \in \mathcal{N}_g}|\mathcal{N}_i|)^2}
\mathbb{E}_{|\bm t_i^u} \left\|\sum_{i \in \mathcal{N}_g}\!\! \frac{1}{t_{i}^{u}} \!\!\sum_{l=0}^{t_{i}^{u}-1}\!\!\sum_{k \in \mathcal{N}_i}\!\! \Tilde{g}_{i,k}\!\left(\!x_{i}^{u,l}\!\right)\right\|^2 \nonumber \\
&+ \alpha  \mathbb{E}_{|\bm t_i^u}\langle\nabla f\left(x^{u}\right),  \frac{-1}{\sum_{i \in \mathcal{N}_g}|\mathcal{N}_i|}\sum_{i \in \mathcal{N}_g} \frac{1}{t_{i}^{u}} \sum_{l=0}^{t_{i}^{u}-1}\sum_{k \in \mathcal{N}_i} \Tilde{g}_{i,k}\left(x_{i}^{u,l}\right)\rangle. 
\end{align}
For the inner product term above, we have
\begin{align} \label{global_dot_bound}
 &\hspace{-0.3in}\alpha \mathbb{E}_{|\bm t_i^u} \langle \nabla f\left(x^{u}\right),\frac{-1}{\sum_{i \in \mathcal{N}_g}|\mathcal{N}_i|}\sum_{i \in \mathcal{N}_g} \frac{1}{t_{i}^{u}} \sum_{l=0}^{t_{i}^{u}-1}\sum_{k \in \mathcal{N}_i} \Tilde{g}_{i,k}\left(x_{i}^{u,l}\right)\rangle \nonumber \\
  &=
  \frac{\alpha}{2} \left(\mathbb{E}_{|\bm t_i^u}\left\|\nabla f\left(x^{u}\right)-   \frac{1}{\sum_{i \in \mathcal{N}_g}|\mathcal{N}_i|}\sum_{i \in \mathcal{N}_g} \frac{1}{t_{i}^{u}} \sum_{l=0}^{t_{i}^{u}-1}\sum_{k \in \mathcal{N}_i}\nabla F_{i,k}\left(x_{i}^{u,l}\right)\right\|^2-\mathbb{E}_{|\bm t_i^u}\left\|\nabla f\left(x^{u}\right)\right\|^2\right.\nonumber\\ 
  &\hspace{1.5in}\left.\quad \qquad -\mathbb{E}_{|\bm t_i^u}\left\|\frac{1}{\sum_{i \in \mathcal{N}_g}|\mathcal{N}_i|}\sum_{i \in \mathcal{N}_g} \frac{1}{t_{i}^{u}} \sum_{l=0}^{t_{i}^{u}-1}\sum_{k \in \mathcal{N}_i}\nabla F_{i,k}\left(x_{i}^{u,l}\right)\right\|^2\right).
\end{align}
For the second term in \eqref{global_smooth_bound}, we have
\begin{align} \label{global_sgd_bound}
\mathbb{E}_{|\bm t_i^u}\left\|\sum_{i \in \mathcal{N}_g} \frac{1}{t_{i}^{u}} \sum_{l=0}^{t_{i}^{u}-1}\sum_{k \in \mathcal{N}_i} \Tilde{g}_{i,k}\left(x_{i}^{u,l}\right)\right\|^2 
   \!\!\leq|\mathcal{N}_g|\!\!\sum_{i \in \mathcal{N}_g }\! \!|\mathcal{N}_i|^2\sigma^2\!\!+ \!\mathbb{E}_{|\bm t_i^u}\!\!\left\| \sum_{i \in \mathcal{N}_g} \!\!\frac{1}{t_{i}^{u}} \!\!\sum_{l=0}^{t_{i}^{u}-1}\!\!\sum_{k \in \mathcal{N}_i} \!\!\nabla F_{i,k}\!\left(x_{i}^{u,l}\!\right)\right\|^2.
\end{align}
Substituting \eqref{global_dot_bound} and \eqref{global_sgd_bound} into \eqref{global_smooth_bound} and choosing $0 < \alpha \leq \frac{1}{L}$  yields
\vspace{-.2in}
\begin{align}  \label{global_smooth_bound_2}
\hspace{-0.5in}\mathbb{E}_{|\bm t_i^u}f\left(x^{u+1}\right) &\leq \mathbb{E}_{|\bm t_i^u}f\left(x^{u}\right)+ 
  \frac{\alpha}{2} 
\mathbb{E}_{|\bm t_i^u}\left\|\nabla f\!\left(x^{u} \right)-   \frac{1}{\underset{i \in \mathcal{N}_g}{\sum}|\mathcal{N}_i|}\sum_{i \in \mathcal{N}_g} \!\!\frac{1}{t_{i}^{u}} \sum_{l=0}^{t_{i}^{u}-1}\!\sum_{k \in \mathcal{N}_i}\!\!\nabla F_{i,k}\left(x_{i}^{u,l}\right)\right\|^2 \nonumber \\
  &\hspace{1.5in}~~-\frac{\alpha}{2}\mathbb{E}_{|\bm t_i^u}\left\|\nabla f \left(x^{u}\right)\right\|^2 + \frac{\alpha^2 L |\mathcal{N}_g| \sum_{i \in \mathcal{N}_g} |\mathcal{N}_i|^2\sigma^2}{2\left(\sum_{i \in \mathcal{N}_g}|\mathcal{N}_i|\right)^2}.
\end{align}
Regarding the second term in \eqref{global_smooth_bound_2}, although the division by $t_i^{u}$ fixes the bias issue of the cumulative gradient at the GPS, it does not make it not coincide with its theoretical definition in \eqref{global_loss}. Hence, different from the analogous step in \eqref{eq_nameless} in the proof of Theorem~\ref{CA_Group}, the term above requires more mathematical manipulations. Towards that end, we bound it as follows:
\begin{align} \label{e_bound2}
&\mathbb{E}_{|\bm t_i^u}\left\|\nabla f\left(x^{u}\right)-   \frac{1}{\underset{i \in \mathcal{N}_g}{\sum}|\mathcal{N}_i|}\sum_{i \in \mathcal{N}_g} \frac{1}{t_{i}^{u}} \sum_{l=0}^{t_{i}^{u}-1}\sum_{k \in \mathcal{N}_i}\nabla F_{i,k}\left(x_{i}^{u,l}\right)\right\|^2 \nonumber \\
&\leq \frac{|\mathcal{N}_{g}|}{\left(\sum_{i \in \mathcal{N}_g}|\mathcal{N}_i|\right)^2} \sum_{i \in \mathcal{N}_g} |\mathcal{N}_{i}|   \sum_{k \in \mathcal{N}_i} \mathbb{E}_{|\bm t_i^u}\left\| \nabla F_{i,k}\left(x^{u}\right)
    - \frac{1}{t_{i}^{u}} \sum_{l=0}^{t_{i}^{u}-1}\nabla F_{i,k}\left(x_{i}^{u,l}\right)\right\|^2\nonumber \\ 
    &\leq  \frac{2 L^2 |\mathcal{N}_{g}|}{\left(\sum_{i \in \mathcal{N}_g}|\mathcal{N}_i|\right)^2} \sum_{i \in \mathcal{N}_g} \!|\mathcal{N}_{i}| \sum_{k \in \mathcal{N}_i} \frac{1}{t_i^u}  \sum_{l=0}^{t_i^{u}-1} \mathbb{E}_{|\bm t_i^u} \left\|x^{u}-x_{i}^{u-1,t_{i}^{u-1}}
    \right\|^2+\mathbb{E}_{|\bm t_i^u} \left\|x_{i}^{u-1,t_{i}^{u-1}}-x_{i}^{u,l}\right\|^2.
\end{align}
For the last term above, we have
\begin{align}\label{b_bound}
\hspace{-1.6in}\mathbb{E}_{|\bm t_i^u}\left\|x_{i}^{u-1,t_{i}^{u-1}}-x_{i}^{u,l}\right\|^2
&\leq 2 \mathbb{E}_{|\bm t_i^u}\left\|x_{i}^{u-1,t_{i}^{u-1}}-x^{u}\right\|^2 + 2 \mathbb{E}_{|\bm t_i^u}\left\| x^{u}-x_{i}^{u,l}\right\|^2 \nonumber \\
&\leq 2 \mathbb{E}_{|\bm t_i^u}\left\|x_{i}^{u-1,t_{i}^{u-1}}-x^{u}\right\|^2 +  \frac{2\alpha^2}{|\mathcal{N}_i|}\sum_{k \in \mathcal{N}_i} l^2  G^2.
\end{align}
Next, substituting the bound of (\ref{b_bound}) in (\ref{e_bound2}), 
and then substituting \eqref{e_bound2} in \eqref{global_smooth_bound_2} and rearranging, we get
\vspace{-0.2in}
\begin{align}  \label{global_smooth_bound_3}
    &\mathbb{E}_{|\bm t_i^u}\left\|\nabla f(x^{u})\right\|^2  
\leq  \frac{2}{\alpha} \left(\mathbb{E}_{|\bm t_i^u}f\left(x^{u}\right)-\mathbb{E}_{|\bm t_i^u}f\left(x^{u+1}\right)\right) + \frac{\alpha L |\mathcal{N}_g| \sum_{i \in \mathcal{N}_g} |\mathcal{N}_i|^2\sigma^2}{(\sum_{i \in \mathcal{N}_g}|\mathcal{N}_i|)^2} \nonumber \\
&+ \frac{4 L^2 |\mathcal{N}_{g}|}{\left(\sum_{i \in \mathcal{N}_g}|\mathcal{N}_i|\right)^2} \sum_{i \in \mathcal{N}_g} |\mathcal{N}_{i}| \sum_{k \in \mathcal{N}_i} \frac{1}{t_i^u} \sum_{l=0}^{t_i^{u}-1}  \left(\!3\alpha^2  (t_{i}^{u-1})^2\!+\! \frac{3\alpha^2 |\mathcal{N}_g|}{\left(\sum_{i \in \mathcal{N}_g}|\mathcal{N}_i|\right)^2} \sum_{j \in \mathcal{N}_g}|\mathcal{N}_j|^2G^2+  \alpha^2l^2  G^2 \!\right).
\end{align}
Then, taking the average over global communication rounds $\mathcal{U}$ yields
\begin{align}  
   &\hspace{-.5in}\frac{1}{\mathcal{U}} \mathbb{E}_{|\bm t_i^u}\left\|\nabla f\left(x^{u}\right)\right\|^2  
\leq  \frac{2}{\alpha} \frac{1}{\mathcal{U}}  \left(\mathbb{E}_{|\bm t_i^u}f(x^{1})-\mathbb{E}_{|\bm t_i^u}f\left(x^{\mathcal{U}+1}\right) \right) + \frac{\alpha L |\mathcal{N}_g| \sum_{i \in \mathcal{N}_g} |\mathcal{N}_i|^2\sigma^2}{\left(\sum_{i \in \mathcal{N}_g}|\mathcal{N}_i|\right)^2}\nonumber \\
&\hspace{-.5in}+ \frac{1}{\mathcal{U}} \sum_{u=1}^{\mathcal{U}}  \frac{4 L^2|\mathcal{N}_{g}|}{\left(\sum_{i \in \mathcal{N}_g}|\mathcal{N}_i|\right)^2}\sum_{i \in \mathcal{N}_g} |\mathcal{N}_{i}| \sum_{k \in \mathcal{N}_i} \frac{1}{t_i^u} \sum_{l=0}^{t_i^{u}-1} 3\alpha^2  (t_{i}^{u-1})^2 + \frac{12 \alpha^2 L^2 G^2 |\mathcal{N}_{g}|^2}{\left(\sum_{i \in \mathcal{N}_g}|\mathcal{N}_i|\right)^4}(\sum_{i \in \mathcal{N}_g} |\mathcal{N}_{i}|^2 )^2 \nonumber \\
&\hspace{-.5in}+ \frac{1}{\mathcal{U}} \sum_{u=1}^{\mathcal{U}}\frac{4\alpha^2 L^2  G^2 |\mathcal{N}_{g}|}{\left(\sum_{i \in \mathcal{N}_g}|\mathcal{N}_i|\right)^2} \sum_{i \in \mathcal{N}_g} |\mathcal{N}_{i}| \sum_{k \in \mathcal{N}_i} \frac{1}{t_i^u} \sum_{l=0}^{t_i^{u}-1}  l^2 . 
\end{align}
Direct simplifications of the above expression give the result of the theorem.
\vspace{-0.15in}
\section{Proof Of Corollary ~\ref{corollary}} \label{appE}
\vspace{-0.2in}
By Theorem~\ref{global_convg}, we have shown the convergence rate of the whole setting. Furthermore, bounding the local iteration time, and as a consequence the number of local iterations as stated in \eqref{t_bound1} and \eqref{t_bound2}, one can show that the bound in Theorem~\ref{global_convg} behaves as follows: 
\begin{align} \label{global_bound} 
 \frac{1}{\mathcal{U}} &\sum_{u=1}^{\mathcal{U}} \mathbb{E}_{|\bm t_i^u}\left\|\nabla f(x^{u})\right\|^2  
\leq 
 \frac{2}{\alpha} \frac{1}{\mathcal{U}}  \left(\mathbb{E}f\left(x^{1}\right)-\mathbb{E}f\left(x^{\mathcal{U}+1}\right)\right)  +\frac{12 \alpha^2 L^2 G^2|\mathcal{N}_{g}|^2 }{\left(\sum_{i \in \mathcal{N}_g}|\mathcal{N}_i|\right)^4} ( \sum_{i \in \mathcal{N}_g} |\mathcal{N}_{i}|^2 )^2 \nonumber \\
&+\frac{\alpha L |\mathcal{N}_g| \sum_{i \in \mathcal{N}_g} |\mathcal{N}_i|^2\sigma^2}{\left(\sum_{i \in \mathcal{N}_g}|\mathcal{N}_i|\right)^2}  +\frac{\alpha^2 }{\mathcal{U}} \!\sum_{u=1}^{\mathcal{U}}  \frac{12 L^2 |\mathcal{N}_{g}|}{\left(\sum_{i \in \mathcal{N}_g}|\mathcal{N}_i|\right)^2}\sum_{i \in \mathcal{N}_g} \!|\mathcal{N}_{i}|^2   (t_{i}^{u-1})^2 \nonumber \\
& + \frac{\alpha^2}{\mathcal{U}} \!\sum_{u=1}^{\mathcal{U}}\frac{4 L^2  |\mathcal{N}_{g}|}{\left(\sum_{i \in \mathcal{N}_g}|\mathcal{N}_i|\right)^2}
 \!\sum_{i \in \mathcal{N}_g}   \!\frac{|\mathcal{N}_{i}|^2(t_i^{u}-1)(2t_i^{u}-1) G^2 }{6} \nonumber \\
\leq& \frac{2}{\alpha} \frac{1}{\mathcal{U}}  \left(\mathbb{E}_{|\bm t_i^u}f\left(x^{1}\right)-\mathbb{E}_{|\bm t_i^u}f\left(x^{\mathcal{U}+1}\right)\right)  +\frac{12  L^2 G^2 \alpha^2 |\mathcal{N}_{g}|^2 }{\left(\sum_{i \in \mathcal{N}_g}|\mathcal{N}_i|\right)^4}(\sum_{i \in \mathcal{N}_g} |\mathcal{N}_{i}|^2 )^2 + \frac{\alpha L |\mathcal{N}_g| \sum_{i \in \mathcal{N}_g} |\mathcal{N}_i|^2\sigma^2}{\left(\sum_{i \in \mathcal{N}_g}|\mathcal{N}_i|\right)^2} \nonumber \\
&+  \frac{12 L^2 \alpha^2  |\mathcal{N}_{g}|}{\left(\sum_{i \in \mathcal{N}_g}|\mathcal{N}_i|\right)^2}
 \sum_{i \in \mathcal{N}_g}   (|\mathcal{N}_{i}|t_{i}^{\max})^2
+ \frac{4 L^2 \alpha^2 G^2  |\mathcal{N}_{g}| }{\left(\sum_{i \in \mathcal{N}_g}|\mathcal{N}_i|\right)^2}\sum_{i \in \mathcal{N}_g}   \frac{(|\mathcal{N}_{i}|t_i^{\max})^2}{3}. 
\end{align}
Choosing $\alpha\leq\frac{1}{\sqrt{\mathcal{U}}}$ in the final bound above directly gives the required $\mathcal{O}(\frac{1}{\sqrt{\mathcal{U}}})$ bound.

\ifCLASSOPTIONcaptionsoff
  \newpage
\fi

\bibliographystyle{IEEEtran}
\bibliography{Ali.bib}

\end{document}